\documentclass[10pt,twocolumn,twoside]{IEEEtran}
\usepackage{setspace}
\onehalfspace
\usepackage{graphicx}
\usepackage{amsmath}
\usepackage{nicefrac}
\usepackage{amsfonts}
\usepackage{times}
\usepackage{algorithm}
\usepackage{algorithmic}
\usepackage{amsthm}
\usepackage{subfigure}
\usepackage[usenames]{color}
\newtheorem{theorem}{Theorem}%[section]
\newtheorem{atheorem}{Theorem}[section]

\newtheorem{lemma}[theorem]{Lemma}

\newtheorem{aproposition}[atheorem]{Proposition}
\usepackage{multirow}

\theoremstyle{definition}
\newtheorem{defn}[theorem]{Definition}

\theoremstyle{remark}
\newtheorem{rem}[theorem]{Remark}

\newcommand{\set}{\{\,1, \, \ldots\,,\,\n\,\}}

\def\RR{\mathbb{R}}
\def\CC{\mathbb{C}}

\newcommand{\has}{{{\hat{\as}}}}
\newcommand{\cx}{{\cal X}}

\newcommand{\n}{{\cal C}}
\renewcommand{\baselinestretch}{1.2}

\newcommand{\pk}{{\cal P}_k}
\newcommand{\mm}{{N}}

\newcommand{\cb}[1]{{ {#1}}}

\newcommand{\bs}[1]{{{#1}}}
\newcommand{\m}{${\cal \mm}$}

\newcommand{\yc}{\left\{1,\cdots,\n\right\}}
\newcommand{\G}{{\cal G}}

\newcommand{\ff}{{f}}
\newcommand{\A}{{\Phi}}

\newcommand{\sqm}{\cb{\frac{1}{\sqrt{\mm}}}}

\newcommand{\range}{\mbox{range}}
\newcommand{\PT}{\mathbb{ P}_\lambda}
\newcommand{\nor}{\frac{1}{\sqrt{\mm}}}

\newcommand{\as}{{\alpha}}
\newcommand{\yy}{{f}}
\newcommand{\Ex}{\mathbb{E}}

\newcommand{\lab}{\ell_2/\ell_1}
\newcommand{\lbb}{\ell_1/\ell_1}

\newcommand{\icd}{}
\newcommand{\Xoni}{X_{1\rightarrow i}}
\newcommand{\Xonim}{X_{1\rightarrow (i-1)}}
\newcommand{\Xim}{X_{(i+1)\rightarrow m}}
\newcommand{\xim}{x_{(i+1)\rightarrow m}}
\newcommand{\xoni}{x_{1\rightarrow i}}
\newcommand{\xonim}{x_{1\rightarrow (i-1)}}
\newcommand{\cxoni}{\cx_{1\rightarrow i}}
\newcommand{\cxonim}{\cx_{1\rightarrow (i-1)}}
\newcommand{\cxim}{\cx_{(i+1)\rightarrow m}}
\newcommand{\ignore}[1]{}           
\title{Construction of a Large Class of Deterministic
Sensing Matrices that Satisfy a Statistical Isometry
Property}
  
\author{Robert Calderbank\thanks{The work of R. Calderbank and S. Jafarpour is supported in part by NSF under grant DMS 0701226, by ONR under grant N00173-06-1-G006, and by AFOSR under grant FA9550-05-1-0443},~\IEEEmembership{Fellow,~IEEE},
Stephen Howard,~\IEEEmembership{Member,~IEEE}, and
Sina Jafarpour,~\IEEEmembership{Student Member,~IEEE}}
\singlespace

\begin{document}

\maketitle         
\begin{abstract}
%$~$\\%\\
Compressed Sensing aims to capture attributes of $k$-sparse signals using very few
measurements. In the standard Compressed Sensing paradigm, the $\mm\times \n$ 
measurement matrix $\A$ is required to act as a near isometry on the set
of all $k$-sparse signals (Restricted Isometry Property or RIP). 
If $\A$ satisfies the RIP, then Basis Pursuit or Matching Pursuit recovery algorithms
can be used to recover any $k$-sparse vector $\as$ from the $\mm$ measurements $\A \as$. 
Although it is known that certain probabilistic processes generate $\mm \times \n$
matrices that satisfy RIP with high probability, there is no practical algorithm for verifying whether
a given sensing matrix $\A$ has this property, crucial for the feasibility
of the standard recovery algorithms. %\\
In contrast this paper provides simple criteria 
that guarantee that a deterministic sensing matrix {\icd{satisfying these criteria}} acts as a near isometry on an overwhelming majority of $k$-sparse signals; in particular, most such signals have a unique representation in the
measurement domain. Probability still plays a critical role, but
it enters the signal model rather than the construction of the sensing matrix. An essential element in our construction is that we require the columns of the sensing matrix to form a group under
pointwise multiplication. The construction allows recovery methods for which the expected performance is sub-linear in $\n$, and only quadratic in $\mm$, as compared to the
 super-linear complexity in $\n$ of the Basis Pursuit or Matching Pursuit algorithms;  
the focus on expected performance is more typical of mainstream signal
processing than the worst-case analysis that prevails in standard Compressed Sensing. 
Our framework encompasses many families of
deterministic sensing matrices, including those formed from discrete chirps,
Delsarte-Goethals codes, and extended BCH codes.
\end{abstract}
\section{Introduction and Notations}
The central goal of compressed sensing is to capture attributes of a signal using 
very few
measurements. In most work to date, this broader objective is exemplified by the important special
case in which a $k$-sparse vector $\as \in \mathbb{R}^{\n}$
(with $\n$ large) is to be reconstructed from a small number
$\mm$ of linear measurements with $k \,<\, \mm \,<\, \n$. 
In this problem, the measurement data constitute a vector
$f\,=\,\mm^{-1/2}\,\A\as\,$, where $\A$ is an $\mm \times \n$ matrix 
called the {\em sensing matrix}. Throughout this paper we shall use the notation
$\varphi_j$ for
the $j$-th column of the sensing matrix $\A$; its entries will be denoted by 
$\varphi_j(x)$ (with label $x$ varying from $1$ to $\mm$). In other words,
$\varphi_j(x)$ is the $x$-th row and $j$-th column element of $\A$. \\%\\
The two fundamental questions in compressed sensing are:
%\\%\\ 
%$\bullet$ how to construct suitable sensing matrices $\A$, and\\%\\
%$\bullet$ how to recover efficiently $\as$ from $f$. \\%\\
how to construct suitable sensing matrices $\A$, and how to recover 
$\as$ from $f$ efficiently;
it is also of practical importance to be resilient to
measurement noise and to be able to reconstruct (approximations to) 
$k$-compressible signals, i.e. signals {\icd{that have more than $k$ nonvanishing entries, but}} where only $k$ entries are
significant and the remaining entries are close to zero.
\begin{center}
\vspace{-0.7cm}
\begin{table*}[t]
\caption{Properties of  $k$-sparse reconstruction algorithms that employ random sensing matrices with $\mm$ Rows and $\n$ Columns. The property RIP-1 is the counterpart of RIP for the $\ell_1$ metric and it provides guarantees on the performance of sparse reconstruction algorithms that employ linear programming \cite{BGIST}. Note that explicit construction of the expander graphs requires a large number of measurements, and that more practical alternatives are random sparse matrices which are expanders with high probability.\vspace{0.5cm}}
\label{ralgs}
\begin{center}

\begin{minipage}{\textwidth} %for 2-column document, use 0.5\textwidth 
\begin{center}
\begin{tabular}{|c|c|c|c|c|c|c|c|c|}
 \hline 
Approach & Number of & Complexity & Compressible  & Noise   & RIP\\
 &  Measurements $\mm$ & & Signals & Resilience &
\\
\hline 
Basis Pursuit& $k \log\left(\frac{\n}{k}\right)$ & ${\n}^3$ & Yes & Yes & Yes 
\\ (BP) \cite{CRT1} & & & & &   
\\
\hline
Orthogonal Matching   & $k \log^\alpha(\n)$ & $k^2 \log^\alpha(\n)$ & Yes & No & Yes 
\\ Pursuit (OMP) \cite{GSTV} & & & & &   
\\
\hline
Group Testing \cite{CM}& $k \log^\alpha(\n)$ & $k \log^\alpha(\n)$ & Yes & No & No 
\\ 
\hline
Expanders (Unique &$k \log\left(\frac{\n}{\mm}\right)$&$\n \log\left(\frac{\n}{\mm}\right)$& $\mbox{Yes}^{\circ}$ &$\mbox{Yes}^{\circ}$ &RIP-1
\\
Neighborhood) \cite{Sina} & & & & &
\\ 
\hline
Expanders (BP) \cite{BGIST} & $k \log\left(\frac{\n}{k}\right)$ & ${\n}^3$ & Yes & Yes &   RIP-1
\\
\hline
Expander Matching & $k \log\left(\frac{\n}{k}\right)$ & $\n \log\left(\frac{\n}{k}\right)$ & $\mbox{Yes}^{\circ}$ & $\mbox{Yes}^{\circ}$ &  RIP-1
\\ Pursuit (EMP) \cite{IR} & & & & &   
\\
\hline
Sparse Matching & $k \log\left(\frac{\n}{k}\right)$ & $\n \log\left(\frac{\n}{k}\right)$ & $\mbox{Yes}^{\circ}$ & $\mbox{Yes}^{\circ}$ &  RIP-1
\\ Pursuit (SMP) \cite{IR} & & & & &   
\\
\hline
CoSaMP \cite{NT}& $k \log\left(\frac{\n}{k}\right)$ & $\n k \log\left(\frac{\n}{k}\right)$ & Yes & Yes  & Yes
\\
\hline
SSMP \cite{DM} & $k \log\left(\frac{\n}{k}\right)$ & $\n k \log\left(\frac{\n}{k}\right)$ & Yes & Yes & Yes
\\
\hline 
\end{tabular}
\end{center}
%\\ \\ 
$\circ$ \cite{Sina} provides an algorithm with smaller constants that is easier to implement and analyze, whereas \cite{IR} is able to handle more general noise models.
\end{minipage}
\vspace{-0.5cm}
\end{center}
\end{table*}
\end{center}

The work of Donoho \cite{Donoho} and of Cand\`{e}s, Romberg and Tao 
\cite{CT}, \cite{CRT1}, \cite{CRT2}
provides fundamental insight into the geometry of sensing matrices. 
This geometry is expressed by e.g. the  Restricted Isometry
Property (RIP), formulated by Cand\`{e}s and Tao \cite{CT}: a sensing matrix satisfies
the $k$-Restricted Isometry Property if it acts as a
near isometry on all $k$-sparse vectors; to ensure unique and stable
reconstruction of $k$-sparse vectors, it is sufficient that $\A$ satisfy 
$2k$-RIP. When $\mm/\n$
and/or $k/\mm$ are (very) small, deterministic RIP matrices 
have been constructed using methods from approximation theory \cite{Devore} and coding
theory \cite{IND}. More attention has been paid to probabilistic constructions where the entries of
the sensing matrix are generated by an i.i.d Gaussian or Bernoulli process or from random
Fourier ensembles, in which larger values of $\mm/\n$
and/or $k/\mm$ can be considered. These sensing matrices are known to satisfy the $k$-RIP with high probability
\cite{Donoho}, \cite{CT} and the number $\mm$ of measurements is 
$k \log\frac{\n}{k}\,$.
This is best possible in the
sense that approximation results of Kashin \cite{6} and Glushin \cite{5} imply that 
$\Omega(\,k\,\log\,\frac{\n}{k}\,)$
measurements are required for sparse reconstruction using 
$\ell_1$-minimization methods. 
Constructions
of random sensing matrices of similar size that have the RIP 
but require a smaller degree
of randomness, are given by several approaches including filtering \cite{Nowak}, 
\cite{Tropp} and
expander graphs \cite{GLR}, \cite{BGIST}, \cite{IR}, \cite{Sina}. 

The role of random measurement in
compressive sensing can be viewed as analogous to the role of random coding in Shannon theory. Both provide
worst case performance guarantees in the context of an adversarial signal/error model. Random
sensing matrices are easy to construct, and are $2k$-RIP with high probability.
As in coding theory, this randomness has its drawbacks, briefly described as follows:
\\%\\
$\bullet$ First, efficiency in sampling comes at the
cost of complexity in reconstruction (see Table 1) and at the cost of error in signal
approximation (see Section 5). \\%\\
$\bullet$ Second, storing the
entries of a random sensing matrix may require significant space, in contrast to deterministic matrices where the entries can
often be computed on the fly without requiring any storage. \\%\\
$\bullet$ Third, there is no algorithm for
efficiently verifying whether a sampled sensing matrix satisfies RIP, a condition that is essential for the
recovery guarantees of the Basis Pursuit and Matching Pursuit algorithms on {\em any} sparse
signal. \\%\\
These drawbacks lead us to consider constructions with deterministic sensing matrices, for which the performance is 
guaranteed in expectation only, for $k$-{\icd{sparse}} signals that are random variables, but which do not suffer from the same drawbacks.
The framework presented here provides \\%\\
$\bullet$ easily checkable conditions on
special types of deterministic sensing matrices
guaranteeing 
successful recovery of {\em all but an exponentially small fraction} of $k$-sparse signals; \\%\\
$\bullet$ in many examples, the entries of these matrices can
be computed on the fly without requiring any storage, and\\%\\ 
$\bullet$ recovery algorithms with lower complexities than Basis Pursuit and Matching Pursuit
algorithms.\\%\\
To make this last point more precise, we note that Basis Pursuit and Matching Pursuit
algorithms rely heavily on {\em matrix-vector} multiplication, and are super-linear with respect to $\n$,
the dimension of the data domain. The reconstruction
algorithm for the framework presented here (see Section 5) requires only 
{\em vector-vector} multiplication in the measurement domain; as a
result, its recovery time is only quadratic in the dimension $\mm$ of the measurement domain.
\begin{table*}[t]
\caption{Properties of $k$-sparse reconstruction algorithms that employ deterministic sensing matrices with $\mm$ Rows and $\n$ Columns. Note that for LDPC codes $k\ll \n$. Note also that RIP holds for random matrices where it implies existence of a low-distortion embedding from $\ell_2$ into $\ell_1$. Guruswami et al. \cite{GLR} proved that this property also holds for deterministic sensing matrices constructed from expander codes. It follows from Theorem \ref{maintheorem} in this paper that  sensing matrices based on discrete chirps and Delsarte-Goethals codes satisfy the UStRIP.
}
\label{dalgs}
\begin{center}
\begin{minipage}{\textwidth} %for 2-column document, use 0.5\textwidth 
\begin{center}
\begin{tabular}{|c|c|c|c|c|c|c|c|c|}
 \hline 
Approach & Number of & Complexity & Compressible  & Noise   & RIP\\
 &  Measurements $\mm$ & & Signals & Resilience &
\\
\hline 
Low Density & & & & &  
\\ Parity Check  Codes & $k \log \n$ & $\n \log \n$ & Yes & Yes & No  
\\ (LDPC) \cite{LDPC}& & & & &   
\\ 
\hline 
Low Density & & & & &  
\\ Parity Check  Codes & $k \log\left(\frac{\n}{k}\right)$ & $\n $ & Yes & Yes & No  
\\ (LDPC) \cite{apt}& & & & &   
\\
\hline
Reed-Solomon & $k$ & $k^2$ & No & No &  No 
\\ codes \cite{AT}& & & & &   
\\
\hline
Explicit Construction of & $\n$ & $\n$ & Yes & Yes &  No 
\\  Expander Graphs \cite{xh}& & & & &   
\\
\hline
Embedding  $\ell_2$ into & $k (\log \n)^{\alpha \log \log \n}$ &${\n}^3$ & Yes & No & No 
\\$\ell_1$ (BP) \cite{GLR} & & & & &   
\\
\hline
Extractors \cite{IND} & $k{\n}^{o(1)}$ & $k{\n}^{o(1)}\log(\n)$ & No & No & No
\\
\hline
Discrete chirps \cite{LHSC}  & $\sqrt{\n}$ & $k \mm$ $\log$ $\mm$  & Yes & Yes &  UStRIP
\\
\hline
\cb{Delsarte-Goethals codes}  & \cb{$k \log\n$} &  \cb{$k^2 \log^{2+o(1)}\n$} & Yes & Yes & UStRIP
\\ \cb{This Paper, \cite{HSC,quad}}& & & & &   
\\
\hline 
\end{tabular}
\end{center}
\end{minipage}
\end{center}
\end{table*}

We suggest that the role of the deterministic measurement matrices presented here for compressive sensing is
analogous to the role of structured codes in communications practice:  in both cases 
fast encoding and decoding algorithms are emphasized, and typical rather than worst
case performance is optimized. We are not the only ones seeking inspiration in coding theory 
to construct deterministic matrics for compressed sensing; Table 2 gives an overview
of approaches in the literature that employ deterministic sensing matrices, several of
which are based on linear codes (cf. \cite{LDPC} and \cite{AT}) and provide
expected-case rather than worst-case performance guarantees.
It is important to note (see Table 2) that although the use of linear codes makes fast algorithms possible for sparse reconstruction, these are not always resilient to noise.
Such non-resilience ma{\icd{n}}ifests itself in e.g. Reed-Solomon (RS) constructions \cite{AT}; 
the RS reconstruction algorithm (the roots of which go back to 1795! -- see \cite{prony},
\cite{Wolf}) uses the input data to construct
an error-locator polynomial; the roots of this polynomial identify the signals appearing in the
sparse superposition. Because the correspondence between the coefficients of a polynomial and
its roots is not well conditioned, it is very difficult to deal with compressible signals and
noisy measurements in RS-based approaches.

Because we will be interested in expected-case performance only, we need not impose
RIP; we shall instead work with the weaker Statistical Restricted Isometry Property.
More precisely, we define
\begin{defn} {\bf{(}}$\boldsymbol{(k,\epsilon,\delta)}${\bf{-StRIP matrix)}}\\%\\
An $\,\mm \times \n$ (sensing) matrix $\A$ is said to be a
$(k,\epsilon,\delta)$-Statistical Restricted Isometry Property matrix
[abbreviated $(k,\epsilon,\delta)$-StRIP matrix] if, for $k$-sparse
vectors $\as \in \RR^{\n}$, the inequalities
\begin{equation}
(1-\epsilon)\,\|\as\|^2 \,\leq \,\left|\!\left|\,\frac{1}{\sqrt{\mm}}\A\as \, \right|\!\right|^2 \,\leq\,(1+\epsilon)\,\|\as\|^2 \,,
\label{strip}
\end{equation}
hold with probability exceeding $1-\delta\,$ (with respect to a uniform distribution of the vectors
$\as$ among all $k$-sparse vectors in $\RR^{\n}$ of the same norm).\footnote{Throughout the paper norms without subscript denote $\ell_2$-norms}.
\end{defn}
There is a slight wrinkle in that, unlike the simple RIP case, StRIP does not automatically imply unique reconstruction, not even with high probability. 
If an $\,\mm \times \n$ matrix $\A$ is $(2k,\epsilon,\delta)-$StRIP, then, given a $k$-sparse vector $\as$, it does follow that $\A$ maps any other randomly picked $k$-sparse signal $\beta$ to a 
different image, i.e. $\A\,\as \,\neq \, \A \,\beta$, with probability exceeding $1-\delta\,$ (with respect to the random choice of $\beta$). This does not mean,
however, that uniqueness is guaranteed with high probability: requiring that 
the measure of 
$\{\,\as \in \RR^{\n} \,;\,\as\, \mbox{ is } k\mbox{-sparse and there is a different } k-\mbox{sparse } \,\beta \in \RR^{\n} \,\mbox{ for which }\, \A\,\as\,=\,\A\,\beta\, \}$
 be small, is a more st{\icd{r}}ingent requirement  than  that the measure of
$ \{\,\beta \in \RR^{\n} \,;\,\beta \,\neq\, \as
\,\mbox{ and }\, \A\,\as\,=\,\A\,\beta\, \}\,$ be small for 
all $k$-sparse $\as$. For this reason, we also introduce the following definition:
\begin{defn} {\bf{(}}$\boldsymbol{(k,\epsilon,\delta)}${\bf{-UStRIP matrix)}}\\%\\
An $\,\mm \times \n$ (sensing) matrix $\A$ is said to be a
$(k,\epsilon,\delta)$-Uniqueness-guaranteed Statistical Restricted Isometry Property matrix
[abbreviated $(k,\epsilon,\delta)-$UStRIP matrix] if $\A$ is a 
$(k,\epsilon,\delta)$-StRIP matrix, and 
\[
\{\,\beta \in \RR^{\n} \,;\,\, \A\,\as\,=\,\A\,\beta \} \,=\,
\{ \, \as \,\}
\]
with probability exceeding $1-\delta\,$ (with respect to a uniform distribution of the vectors
$\as$ among all $k$-sparse vectors in $\RR^{\n}$ of the same norm).
\end{defn}

Again, we are not the first to propose a weaker version of RIP that permits 
the construction of deterministic sensing matrices. {\icd{The construction by
Guruswami et al. in \cite{GLR} can be viewed as another instance of a weakening of RIP, in the following different direction.}}  RIP implies that
$\A$ defines a low-distortion $\ell_2$-$\ell_1$-embedding that plays a 
crucial role
in the proofs of \cite{Donoho},
\cite{CT}, \cite{CRT1}, \cite{CRT2}. 
In \cite{GLR}, Guruswami et al. prove that this $\ell_2$-$\ell_1$-embedding property 
also holds
for deterministic sensing matrices constructed from expander codes. These matrices satisfy an ``almost Euclidean null space property'' property, that is for any $\alpha$ in the null space of $\Phi$,  $\frac{\sqrt{N}\|\alpha\|_2}{\|\alpha\|_1}$ is bounded by a constant ; this is their main 
tool to obtain the results reported in Table 2.

In this paper we formulate simple design rules, imposing that the columns
of the sensing matrix form a group under pointwise multiplication, that all row sums vanish, that different rows are orthogonal, 
%\footnote{\mmote that this point, requiring that different rows of $\A$ are orthogonal, is equivalent to requiring that the columns
%of $\A$ constitute a tight frame} 
and requiring a simple upper bound on the absolute value of any column sum (other than the multiplicative identity). 
The properties we require
are satified by a large class of matrices constructed by exponentiating codewords from a linear
code; several examples are given in Section 2. In Sections 3, we show that our relatively weak design rules are suficient to guarantee that $\A$ is UStRIP, provided the parameters satisfy certain constraints. 
The group property makes it possible to avoid intricate
combinatorial reasoning about coherence of collections of mutually unbiased bases (cf. \cite{GH}). Section 4 applies our results to the case where the sensing matrix is formed by taking random rows of the FFT matrix. In Section 5 we
emphasize a particular family of constructions involving subcodes of the second order Reed-Muller code; in this case 
codewords correspond to multivariable quadratic functions defined over the binary field or the
integers modulo 4. Section~\ref{sec:noise} provides a discussion regarding the noise resilience.

\section{StRIP-able: Basic Definitions, with Several Examples}
\label{sec:matrices}
In this section we formulate three basic %StRIP  %icd: removed StRIP here
conditions and give examples of deterministic sensing matrices $\A$ with $\mm$ rows and $\n$ columns
that satisfy these conditions. Note that throughout the paper, we shall assume (without stating this again explicitly) that $\A$ has no repeated columns.

\begin{defn}
An $\mm \times \n -$matrix $\A$ is said to be $\eta-${\bf{StRIP-able}}, where
$\eta$ satisfies $0\,<\,\eta\,\leq\,1$, if the following three conditions are satisfied:
\begin{itemize}
\item {\bf{(St1)}} The rows of $\A$ are orthogonal, and all the row sums are zero. i.e.
\begin{eqnarray}
&&\sum_{j=1}^{\n}\, \varphi_j(x)\,\overline{\varphi_j(y)} \,=\, 0 \mbox{ if } x\,\neq\,y \label{st1_a}\\%\\
&&\sum_{j=1}^{\n}\, \varphi_j(x)\,=\,0 ~, \mbox{ for all } x ~.\label{st1_b}
\end{eqnarray}
\item {\bf{(St2)}} The columns of $\A$ form a group under ``pointwise multiplication'', defined as follows 
\begin{eqnarray}
&&\mbox{for all } j,\,j' \in \{1,\ldots,\n\},\nonumber\\&&\mbox{there exists a }
j'' \in \{1,\ldots,\n\} \mbox{ such that}\nonumber\\%\\
&&\quad\mbox{for all } 
x\,:~ \varphi_j(x) \,\varphi_{j'}(x)\,=\,\varphi_{j''}(x)~.
\label{st2}
\end{eqnarray}
In particular, there is one column of $\A$ for which all the entries are 1, and that acts as a unit for this group operation; this column will be denoted by $\mathbf{1}$.
Without loss of generality, we will assume the columns of $\A$ are ordered so that
$\varphi_1\,=\,\mathbf{1}$, i.e. $\varphi_1(x)\,=\,1$ for all $x$.
\item {\bf{(St3)}} 
For all $j \in \{2,\ldots,\n\}$, 
\begin{equation}
\left|\sum_x\varphi_j(x) \right|^2 \, \leq \, \mm^{2-\eta}~.
\label{st3}
\end{equation}
\end{itemize}
\end{defn}

{\em Remarks}\\%\\
%\begin{enumerate}
%\item
1. Condition (\ref{st3}) applies to all columns {\em except} the first column (i.e. the column  which consists of all ones).\\%\\
%\item 
2. The justification of the name {\em{StRIP-able}} will be given in the next section.\\%\\
%\item
3. When the value of $\eta$ in (\ref{st3}) does not play a special role, we just don't spell it out explicitly, and simply call $\A$ StRIP-able.
%\end{enumerate}

The conditions (\ref{st1_a}-\ref{st3}) have the following immediate consequences:

\begin{lemma}
If the matrix $\A$ satisfies {\rm{(\ref{st2})}}, then $\left|\varphi_j(x)\right|\,=\,1$,
for all $j$ and all $x$.
\label{ll1}
\end{lemma}
\begin{proof}
For every $x$,  
$\left(\varphi_j(x)\right)_{j \in \{1,\ldots,\n\}}$ is a group of complex numbers under multiplication; all finite groups of this type consist of unimodular numbers.
\end{proof}

\begin{lemma}
If the matrix $\A$ satisfies {\rm{(\ref{st2})}} , then the collection of columns of $\A$ is closed under complex conjugation, i.e.
$\mbox{for all } j \in \{1,\ldots,\n\},\nonumber$ $\mbox{there exists a }
j'  \in \{1,\ldots,\n\}$ \begin{equation}\mbox{ such that, for all } 
x, \quad\varphi_{j'}(x) \,=\,\overline{\varphi_{j}(x)}~.
\label{st4}
\end{equation}
\end{lemma}
\begin{proof}
Pick $j \in \{1,\ldots,\n\}$. Since the columns of $\A$ form a group under pointwise multiplication, there is some $j'  \in \{1,\ldots,\n\}$ such that 
$\varphi_{j'}$ is the inverse of $\varphi_{j}$ for this group operation. Using Lemma 
\ref{ll1}, we have then,
for all $x$,
$\varphi_{j'}(x)\,=\,\left[\varphi_{j}(x)\right]^{-1}\,=\,\overline{\varphi_{j}(x)}$.
\end{proof}

\begin{lemma}
If the matrix $\A$ satisfies {\rm{(\ref{st1_a})}} , {\rm{(\ref{st1_b})}} and {\rm{(\ref{st2})}} , then the normalized columns 
$\left(\,\mm^{-1/2}\,\varphi_j\right)_{j \in \{1,\ldots,\n\}}\,$ form a tight frame in 
$\CC^{\mm}$, with redundancy $\,\n/\mm$.
\label{ll2}
\end{lemma}
\begin{proof}
By Lemma \ref{ll1} and (\ref{st1_a}), we have 
\[
\left(\A\, \A^{\dagger}\,\right)_{x,y}\,=\,\sum_{j=1}^{\n}\,\varphi_j(x)\,\overline{\varphi_j(y)}\,=\,\n\, \delta_{x,y}\]$\mbox{ i.e. } 
\A\, \A^{\dagger}\,=\,\n\, \mbox{I}_\mm,
$
~so that, for any vector $v \in \CC^\mm$, 
\[
\sum_{j=1}^{\n}\,\left| \langle\,v,\,\varphi_j \rangle\right|^2\,=\,
v\,\A\, \A^{\dagger}\,v^{\dagger}\,=\,\n\,\|v\|^2.
\]
\end{proof}

\begin{lemma}\label{inner}
If the matrix $\A$ satisfies {\rm{(\ref{st2})}} , then the inner product of two columns 
$\varphi_j$ and $\varphi_{j'}$, defined as 
$\varphi_j\cdot\varphi_{j'}\,:=\,\sum_{x}\,\varphi_j(x)\,\overline{\varphi_{j'}(x)}$ ,
equals $\mm$ if and only if $j\,=\,j'$.
\end{lemma}
\begin{proof} $~$ \\%\\ 
If $j\,=\,j'$, we obviously have $\varphi_j\cdot\varphi_{j'}\,=\,\mm$, by Lemma \ref{ll1}.\\%\\
If $\varphi_j\cdot\varphi_{j'}\,=\,\mm$, then we have, by Cauchy-Schwarz,
\[
\mm\,=\,\varphi_j\cdot\varphi_{j'}\,\leq\,
\left|\,\varphi_j\cdot\varphi_{j'}\,\right|\,\leq\,\|\varphi_j\|\,\|\varphi_{j'}\|\,=\,\mm~,
\]
implying that in this instance the Cauchy-Schwarz inequality must be an equality, so that
$\varphi_{j'}$ must be some multiple of $\varphi_{j}$. Since 
$\mm\,=\,\varphi_j\cdot\varphi_{j'}$, the multiplication factor must equal 1, so that
$\varphi_j\,=\,\varphi_{j'}\,$. Since $\A$ has no repeated columns, $j\,=\,j'$ follows.
\end{proof}

We shall prove that StRIP-able matrices have (as their name already announces) a Restricted Isometry Property in a Statistical sense, provided the different parameters satisfy certain constraints, which will be made clear and explicit in the next section. Before we embark on that mathematical analysis, we show
that there are many examples of StRIP-able matrices.

\subsection{Discrete Chirp Sensing Matrices}
Let $p$ be a prime and let $\omega$ be a primitive (complex) $p^{\mbox{th}}$ root of unity. A length $p$ chirp signal takes the form
$$\varphi_{mp+r}(x)=\omega^r \omega^{mx+rx^2}\mbox{   where  }x=0,1,\cdots,p-1.$$
Here $m$ is the \textit{base frequency} and $r$ is the \textit{chirp rate}. 
{\icd{Consider now the family of c}}hirp signals 
{\icd{$\left(\varphi_{mp+r}\right)$}} where ${r,m=0,1,\ldots,p-1}$; the ``extra'' phase factor (usually not present in chirps) ensures that the row sums $\sum_{\ell=0}^{p^2-1}\varphi_{\ell}(x)$ vanish for all $x$. It is easy to check that this family
satisf{\icd{ies}}~(St1),~(St2), and~(St3) \cite{LHSC}. {\icd{For the corresponding sensing matrix $\Phi$,}} Applebaum et al. \cite{LHSC} have analyzed an algorithm for sparse reconstruction that exploits the efficiency of the FFT in each of two steps: the first to recover the chirp rate and the second to recover the base frequency. The Gerschgorin Circle Theorem \cite{hj} is used to prove that the RIP holds for sets of $\frac{(\sqrt{p}+1)}{2}$ columns. Numerical experiments reported in \cite{LHSC} compare the eigenvalues of deterministic chirp sensing matrices with those of random Gaussian sensing matrices{\icd{. The singular values of restrictions to $k$-dimensional subspaces of $N \times \n$ random Gaussian sensing matrices have a gaussian distribution, with mean $\mu_{N,\n,k}$ and standard deviation 
$\sigma_{N,\n,k}$; the experiments show that, for the same values of $N$, $\n$ and $k$, the singular values of restrictions of deterministic chirp sensing matrices have a similar spread}}{ 
{\icd{around a central value $\mu \in (\mu_{N,\n,k},1)$ that is closer to 1; in fact, the experiments suggest that  $\mu-\mu_{N,\n,k}>\sigma_{N,\n,k}$.}}
 
\subsection{Kerdock, Delsarte-Goethals and Second Order Reed Muller Sensing Matrices}
{\icd{In our construction of deterministic sensing matrices based on Kerdock, Delsarte-Goethals and second order Reed Muller codes, we start by picking an odd number $m$. The {\icd{$2^m$}} rows of the sensing matrix $\A$ are indexed by {\icd{the}} binary $m$-tuples $x$, and the $2^{(r+2)m}$ columns are indexed by {\icd{the}} pairs $P,b$, where $P$ is an $m\times m$ binary symmetric matrix in the Delsarte-Goethals set $DG(m,r)$, and $b$ is a binary $m$-tuple. The entry $\varphi_{P,b}(x)$ is given by 
\begin{equation}
\label{kerdock}
\varphi_{P,b}(x)=i^{wt(d_P)+2wt(b)}i^{xPx^\top+2bx^\top}
\end{equation}
where $d_p$ denotes the main diagonal of $P$, and $wt$ denotes the {\it Hamming weight} (the number of $1$s in the binary vector). Note that all arithmetic in the exp{\icd{ressions}} $xPx^\top+2bx^\top$ {\icd{and $wt(d_P)+2wt(b)$ takes}} place in the ring of integers modulo $4${\icd{, since they appear only as exponents for $i$. Given $P,b$ the vector $xPx^\top +2bx^\top$ is a codeword in the Delsarte-Goethals code (defined over the ring of integers modulo $4$) For a fixed}} matrix $P$, the $2^m$ columns $\varphi_{P,b}~,~b\in \mathbb{F}_2^m$ form an orthonormal basis $\Gamma_P$ that {\icd{ can also be}} obtained by postmultiplying the Walsh-Hadamard basis by the unitary transformation $\mbox{diag}\left[ i^{xPx^\top}\right]$.

The Delsarte-Goethals set $DG(m,r)$ is a binary vector space containing $2^{(r+1)m}$ binary symmetric matrices with the property that the difference of any two distinct matrices has rank at least $m-2r$ (See \cite{H}). The Delsarte-Goethals sets are nested
$$DG(m,0)\subset DG(m,1) \subset \cdots \subset DG(m,\nicefrac{(m-1)}{2}).$$

The first set $DG(m,0)$ is the classical Kerdock set, and the last set $DG(m,\nicefrac{(m-1)}{2})$ is the set of all binary symmetric matrices. The $r^{th}$ Delsarte-Goethals sensing matrix is determined by $DG(m,r)$ and has $\mm=2^m$ rows and $\n=2^{(r+2)m}$ columns. The initial phase in~(\ref{kerdock}) is chosen so that the Delsarte-Goethals sensing matrices satisfy~(St1) and (St2). (See Appendix A). 

Coherence between orthonormal bases $\Gamma_P$ and $\Gamma_Q$ indexed by binary symmetric matrices $P$ and $Q$ is determined by the rank $R$ of the binary matrix $P\oplus Q$ (See Appendix A). Any vector in one of the orthonormal bases has inner product of absolute value $2^{-\nicefrac{R}{2}}$ with $2^R$ vectors in the other basis and is orthogonal to the remaining basis vectors. {\icd{T}}he column sums in th{\icd{is}} $r^{th}$ Delsarte-Goethals sensing matrix satisfy
$$\left| \sum_x \varphi_{P,b}(x)\right|^2= 0~\mbox{or}~\mm^{2-\nicefrac{r}{m}} ,$$
{\icd{so that condition~(St3) is trivially satisfied.}}
Details are provided in Appendix A{\icd{;}} we refer the interested reader to \cite{K,DG,H} and Chapter $15$ of \cite{MS} for more information about subcodes of the second order Reed-Muller code.
\subsection{BCH Sensing Matrices}
The Carlitz- Uchiyama Bounds (See Chapter $9$ of \cite{MS}) impl{\icd{y}} that the interval $$\left[ 2^{m-1}-(t-1)2^{\nicefrac{m}{2}}, 2^{m-1}+(t-1)2^{\nicefrac{m}{2}}\right]$$ contains all non-zero weights in the dual of the extended binary BCH code $BCH(m,t)$  of length $\mm=2^m$ and designed distance $e=2t+1$, with the exception of $wt(\boldsymbol{1})=\mm$. Setting $BCH(m,t)^\bot=\langle \boldsymbol{1} \rangle \oplus C_{m,t}$, the columns of the $t^{th}$ BCH sensing matrix are obtained by exponentiating the codewords in $C_{m,t}.$ The column determined by the codeword $c=(c_j)$ is given by
$$ \varphi_c(j)=(-1)^{bc^\top} (-1)^{c_j}~,~\mbox{where }j=0,1,\cdots,2^m-1,$$
and where $b$ is any vector not orthogonal to $C_{m,t}$. Conditions (St1) and (St2) hold by construction and 
\begin{eqnarray}\nonumber \left| (-1)^{bc^\top} \sum_{j=0}^{2^m-1} (-1)^{c_j}\right| ^2 &=& \left|\mm-2wt_H(c)\right|^2\\ \nonumber &\leq& \left [ 2 (t-1)2^{\nicefrac{m}{2}}\right]^2\end{eqnarray}
so that (St3) holds. These sensing matrices have been analyzed by Ailon and Liberty \cite{AL}.

In the binary case, the column sums take the form $N-2w$ where $w$ is the Hamming weight of the exponentiated codeword, and a similar interpretation is possible for codes that are linear over the ring of integers modulo 4 (see \cite{H}). Property (St3) connects the Hamming geometry of the code domain, as captured by the weight enumerator of the code, with the geometry of the complex domain.

\section{Implications for Deterministic {\icd{StRIP-able}} Sensing Matrices: Main Result}

In this section we prove our main result, namely 
that if $\A$ satisfies (St1), (St2) and (St3), then $\A$ is
%if $\A$ is a $\mm \times \n$ 
%StRIP-able sensing matrices 
%with constant $\eta$, then, under certain conditions on $\mm$, $k$, $\n$ and $\eta%$, $\A$ has the 
UStRIP, under certain {\icd{fairly weak}} conditions on the parameters.  More precisely,
\begin{theorem}\label{maintheorem}
Suppose the $\mm \times \n$ matrix  $\A$ is $\eta$-StRIP-able, and 
suppose $k \,<\,1\,+\, (\n\,-\,1)\,\epsilon\,$
{\icd{ and $\eta>1/2$}}. 
Then there exists a constant $c$ such that,
%\begin{enumerate}
%\item
if %$k^2 \,<\,\n$ and 
$\mm \,\geq\,\left(c\,\frac{k \, \log\n}{\epsilon^2}\right)^{\frac{1}{\eta}}$, then $\A$ is $(k,\epsilon,\delta)$-UStRIP with 
$\delta\,:=\,2\exp\left[\,-\,\frac{{\icd{[\epsilon-(k-1)/(\n-1)]}}^2\,\mm^{\eta}}{8\, k}\,\right]$.
%\item
%If $k^2 \,\geq\,\n$ and $\mm \,\geq\,c\,\frac{k^{1+\mbox{\rm{\footnotesize{o}}}(1)} \, \log\n}{\epsilon^2}$, then $\A$ is $(k,\epsilon,\delta)$-UStRIP with $\delta\,:=\,\,2\exp\left[\,-\,\frac{\epsilon^2\,\mm^{\eta}}{2\, k}\,\right]$.
%\end{enumerate}
\label{mainthm}
\end{theorem}

%{\em Remark}\\%\\
%If $\A$ is $1$-StRIP-able, then similar conclusions hold,
%with minor changes to the inequalities.

{\icd{The proof of Theorem \ref{mainthm} has two parts: we shall first, in Section 3.1, prove that $\A$ is StRIP;
when this is established we turn our attention to proving UStRIP in Section 3.2.

\subsection{Proving StRIP}}}

{\icd{\subsection*{3.1.1 $~$ Setting up the Framework} }}
%$\boldsymbol{(k,\epsilon,\delta)}$-StRIP Implications}
It will be convenient to decompose the random process generating the vectors $\as$ as follows: first pick (randomly) the {\em indices} of the nonzero {\icd{entries}} of $\as$,
and then the
{\icd{\em{values}}}
of those entries. 
For the first step, we pick a random permutation $\pi\,\doteq\,$
{\icd{$\left(\pi_j\right)$}}{\icd{$_{j\in \set}$ }} of 
$\set$;{\icd{ the $k$ numbers $\pi_1,\ldots,\pi_k$ will then be the indices of the non-vanishing}} entries of $\as$. Next, we pick $k$ random values $\as_1,\ldots,\as_k$; these will be the non-zero {\icd{values of the}} entries of the vector $\as$. 
Computing expectations with respect to $\as$ can be decomposed likewise; when we average over all possible choices of $\pi$, but not yet over the values of the random variables $\as_1,\ldots,\as_k$, we shall denote such expectations by $\mathbb{E}_{\pi}$, adding a subscript. We start by proving the following 

\begin{lemma}\label{rachel}
For $\pi$, $\A$, $\as$ as described above and $f\,:=\,\mm^{-1/2}\,\A\,\as$, we have
\begin{eqnarray}\nonumber
\left(\,1\,-\,\frac{k-1}{\n-1} \,\right)\,\, \|\,\as\,\|^2
\,&\leq&\,\mathbb{E}_{\pi}\left[\,\|f\|^2\,\right]\,\\\nonumber&\leq&\,
\left(\,1\,+\,\frac{1}{\n-1} \, \right) \, \|\,\as\,\|^2~.
\end{eqnarray}
\label{lemma3_2}
\end{lemma}
\begin{proof}
With the notations introduced above, the entries of 
$f\,:=\,\mm^{-1/2}\,\A\,\as$ are given by \\%\\
$f(x)\,:=\,\mm^{-1/2}\,\sum_{j=1}^k\,\as_j\, \varphi_{\pi_j}(x)\,$. We have then
\begin{eqnarray}
\|\,f\,\|^2\,&=&\,\sum_{x=1}^\mm\,|\,f(x)|^2\,\nonumber\\ &=&\,\frac{1}{\mm}\,\sum_{x=1}^\mm\,
\left(\,\sum_{j=1}^k\,|\,\as_j\,|^2\,+\,\Psi(x)\, \right)
\label{eq5}
\end{eqnarray}
where $\Psi(x)\,=\,\sum_{i,j,\,j \neq i} \as_j\,\overline{\as_i}\, \varphi_{\pi_j}(x)\, 
\overline{\varphi_{\pi_i}(x)}\,.$

The first term in (\ref{eq5}) is independent of $\pi$ ; it just equals $\sum_{j=1}^k\,|\as_j\,|^2\,=\,\|\,\as\,\|^2\,$.

For the second term, we have
\begin{eqnarray}\label{eqn6}
& &\mathbb{E}_{\pi} \, 
\left[ \, \sum_x \, \sum_{i,j \mbox{\footnotesize{with} } j\neq i}\, \as_j\, \overline{\as_i}\, \varphi_{\pi_{j}}(x) \, \overline{\varphi_{\pi_{i}}(x)} \,\right]
\\ \,&=&\,\sum_{i,j \mbox{\footnotesize{ with} } j\neq i} \,\as_j \, \overline{\as_i}\, \sum  \mathbb{E}_{\pi}\left[
\sum_x\,\varphi_{\pi_j}(x)\, 
\overline{\varphi_{\pi_i}(x)}\, \right]~.
\nonumber
\end{eqnarray}
By (\ref{st2}) and Lemma \ref{st4}, we have 
$\sum_x\,\varphi_{\ell}(x)\, 
\overline{\varphi_{\ell'}(x)}\,=\,\sum_x\,\varphi_{m}(x)\,$ for some appropriate $m\,:=\,m(\ell,\ell')$;
if $\ell\,\neq \ell'$, then $\overline{\varphi_{\ell'}}\,=\, \left(\varphi_{\ell'}\right)^{-1}\,\neq\, \left(\varphi_{\ell}\right)^{-1}$, so that 
$m(\ell,\ell')\,\neq\,1$. \\%\\
As $\pi$ ranges over all possible permutations of 
$\{1,\ldots,\n\}$, the index $m(\pi_i,\pi_j)$ (with $j\,\neq \,i$) will range (uniformly) over all possible values $2,\ldots,\n$ (i.e. excluding $1$). 
It follows that, for $j\,\neq\,i$,  
\begin{eqnarray}
& &\mathbb{E}_{\pi}\left[
\sum_x\,\varphi_{\pi_j}(x)\, 
\overline{\varphi_{\pi_i}(x)}\, \right]\,\nonumber\\&=&\,(\n-1)^{-1} \,\sum_{\ell \neq 1}\,\sum_x\,\varphi_{\ell}(x)\,\nonumber\\%\\
&=&\,(\n-1)^{-1} \,\sum_x\,(-1)\,=\, -\,\frac{\mm}{\n-1}\,,
\label{eqn7}
\end{eqnarray}
where we have made use of a counting argument in the first equality, and of (\ref{st1_b}) in the second.
It then follows that 
\begin{eqnarray}
& &\mathbb{E}_{\pi}\left[
\,\sum_x \,
\sum_{i,j \mbox{\footnotesize{with}} j\neq i} \,\as_j \, \overline{\as_i}\,\varphi_{\pi_j}(x)\, 
\overline{\varphi_{\pi_i}(x)}\, \right]\,\nonumber\\\nonumber&=&\,-\,\frac{\mm}{\n-1}\,{\sum^k_{i,j; \mbox{\footnotesize{ with }} j\neq i}}\,\as_j \, \overline{\as_i}\,.
\end{eqnarray}
Applying the Cauchy-Schwarz inequality, we obtain
\begin{eqnarray}
0 \, &\leq& \, \sum_{i,j \mbox{ \footnotesize{with} } j\neq i} \,\as_j\, \overline{\as_i}\,
+\, \sum_{j=1}^k\, |\as_j|^2 \,\nonumber\\\nonumber&=&\, \left| \,\sum_{j=1}^k\,\as_j \,\right| ^2 \, \leq \,k\, \sum_{j=1}^k\, |\as_j|^2. 
\end{eqnarray}
\vskip-0.2cm
Combining this with the previous equality gives
\begin{eqnarray}
\nonumber&&-\,\frac{\mm\,(k-1)}{\n-1}\,\|\,\as\,\|^2
\\\nonumber\,&\leq&\, \mathbb{E}_{\pi} \, 
\left[ \, \sum_x \, \sum_{i,j \mbox{\footnotesize{  with }} j\neq i}\, \as_j\, \overline{\as_i}\, \varphi_{\pi_{j}}(x) \, \overline{\varphi_{\pi_{i}}(x)} \,\right]
\,\\\nonumber&\leq&\,\frac{\mm}{\n-1}\,\|\,\as\,\|^2
\end{eqnarray}
It then suffices to substitute this into (\ref{eq5}) to prove the Lemma.
\end{proof}
{\icd{
\begin{rem}
\label{rem_ell}
By using the Cauchy-Schwarz inequality in the last step of the proof of \ref{lemma3_2} we may have sacrificed quite a bit, especially if the non vanishing entries in $\as$ differ appreciably in order of magnitude. Without this step, the final inequality would be
\begin{eqnarray}
\nonumber& &-\,\frac{\mm}{\n-1}\,
\left(\,\|\as\|_{\ell_1}^2\,-\,\|\,\as\,\|^2\,\right)
\,\\\nonumber&\leq&\, \mathbb{E}_{\pi} \, 
\left[ \, \sum_x \, \sum_{i,j \mbox{ \footnotesize{with} } j\neq i}\, \as_j\, \overline{\as_i}\, \varphi_{\pi_{j}}(x) \, \overline{\varphi_{\pi_{i}}(x)} \,\right]
\,\\&\leq&\,\frac{\mm}{\n-1}\,\|\,\as\,\|^2
\label{better}
\end{eqnarray}
\end{rem}
}}
%\cre{
To prove the concentration of $\|{f}\|^2$ around its expected value, we will 
make use of a {\icd{version}} of {\icd{the}} McDiarmid inequality \cite{mcdiarmid} based on concentration of martingale difference random variables with {\icd{\em distinct}} values{\icd{ (as opposed to {\em independent} values for the standard McDiarmid inequality). In what follows, upper case}} letters denote random variables, {\icd{lower case}} letters denote values {\icd{taken on by these random variables.}}
%\cre{
\\%\\
\begin{theorem}[{\icd{Self-Avoiding}} McDiarmid inequality]\label{mac}Let $\cx_1,\cdots,\cx_m$ be probability spaces and define $\cx$ as the probability space of all distinct $m$-tuples\footnote{{\icd{We follow a widespread custom, and denote by the same letter both the set carrying the probability measure, and the probability space [i.e. the triplet (set,$\sigma$-algebra of measurable sets,measure)]. We shall specify which is meant when confusion could be possible.}}}. {\icd{In other words,
the set $\cx$ is the subset of the product set $\mathbb{X}\doteq\cx_1\times\cdots,\times\cx_m$ given by}}
\begin{equation}\label{x2}\cx\doteq\{(t_1,\cdots,t_m)\in \Pi_{i=1}^m \cx_i\mbox{ s.th. $\forall$ }i\neq j~:t_i\neq t_j\}{\icd{;}}\end{equation}
the probability measure on  $\cx$ is just the renormalization (so as to be a probability measure) of the restriction to $\cx$ of the standard product measure on $\mathbb{X}$.\\%\\
Let $h(t_1,\cdots,t_m)$ be a function from {\icd{the set}} $\cx$ to $\mathbb{R}$, such that for any coordinate $i$, given $t_1,\cdots,t_{i-1}$:
\vskip-0.5cm
\begin{eqnarray}
&&\hskip-0.7cm\left|\sup_{u{\icd{\in \cx_i; u \neq t_{n}, n=1\rightarrow i}}}\Ex[h(t_1,\cdots,t_{i-1},u,T_{i+1},\cdots,T_m)]\right.\quad\quad\quad\quad\nonumber\\%\\
&&\hskip-0.7cm-\left.\inf_{l{\icd{\in \cx_i; l \neq t_{n}, n=1\rightarrow i}}}\Ex[h(t_1,\cdots,t_{i-1},l,T_{i+1},\cdots,T_m){\icd{]}}\right|\leq c_i~,\label{concentration2}
\end{eqnarray}
where the expectations are taken over the random variables $T_{i+1}, .., T_m $
(conditioned on taking values that are all different
from each other and from $t_1, \ldots ,t_{i-1}$ as well as $u$ (first expectation) or $l$
(second expectation).
Then for any positive $\gamma$,
\begin{eqnarray}
\nonumber
&&\hskip-1cm\Pr
\left[ 
\left|
h(T_1,\cdots,T_m)- \Ex[h(T_1,\cdots,T_m)]
\right|
\geq\gamma\right]\\\label{result2}
\hskip-1cm&\leq& 2\exp
\left(
\frac{-2\gamma^2}{\sum c_i^2}
\right).
\end{eqnarray}
\label{McDthm}
\end{theorem}

\begin{proof}
 {\icd{See}} Appendix B.
 \end{proof}
% }

{\icd{\subsection*{3.1.2 $~$ Proof of StRIP }}}%in the case $k^2<\n$}}}
We are now ready to start the
\begin{proof} ({\em of the $(k,\epsilon,\delta)$-StRIP property, claimed in Theorem} \ref{mainthm})\\%\\
%for the case $k^2 < \n$ )\\%\\
%
%
%\cre{
\label{strip_proof}
Let $\pk$ denote the set of all $k$-tuples $(\pi_1\,,\cdots\, , \pi_k)$ where $(\pi_1\,, \cdots\, , \pi_\n)$ is a permutation of 
$\{1, 2, \cdots ,\n\}$. It {\icd{follows from the}} definition that all entries of each {\icd{element}} of $\pk$ are distinct.
%}
%\cre{
{\icd{The set $\pk$ is finite; equipped with the counting measure, renormalized so as to have total mass 1, $\pk$ is}} the probability space of the $k$ non-zero entries of the random signal $\as${\icd{: the $(\pi_1,\cdots,\pi_k)$, corresponding to (uniformly) randomly picked permutations $\pi$ of $\set$, are random variables distributed uniformly in $\pk$. For $1\leq i < j \leq k$, we d}}enote {\icd{by $\pi_{i\rightarrow j}$}} the 
{\icd{$(j-i+1)$-tuple}} of random variables ${\icd{(}}\pi_i,\pi_{i+1},\cdots,\pi_j{\icd{)}}$. \\
{\icd{Given values $\alpha_1$, $\alpha_2$, $\ldots$, $\alpha_k$, l}}et ${f}:\pk\rightarrow \mathbb{C}^{\mm}$ {\icd{be}} defined {\icd{by}} ${f}(\pi_1,\cdots,\pi_k)=\nor\sum_{i=1}^k \alpha_i \varphi_{\pi_i}$, and $h:\pk \rightarrow \mathbb{R}$ b{\icd{y}} $h(\pi_1,\cdots,\pi_k)=\|{f}(\pi_1,\cdots,\pi_k)\|^2$. Clearly
\begin{equation}
h(\pi_1,\cdots,\pi_k)=\frac{1}{\mm}\sum_{i,j=1}^k \alpha_i \overline{\alpha_j}\left(\varphi_{\pi_i}\right)^\top \overline{\varphi_{\pi_j}}.
\label{h_eq}
\end{equation}

{\icd{
Our strategy of proof will be the following. We want to upper bound 
$\mbox{Pr}_{\pi}[|\|\ff\|^2 - \|\alpha\|^2| \geq \epsilon \|\alpha\|^2]$. From Lemma \ref{lemma3_2} we know that 
$\Ex_{\pi}[\|\ff\|^2]$ is close to $\|\alpha\|^2$. This suggests that we investigate, for
$\beta >0$, the function $G(\beta)$ defined by
$G(\beta)\doteq\mbox{Pr}_{\pi}[|\|\ff\|^2 -  \Ex_{\pi}[\|\ff\|^2]|]\geq \beta\|\alpha\|^2]=$
$\mbox{Pr}_{\pi}[|h -  \Ex_{\pi}[h]\geq \beta\|\alpha\|^2]$. This last expression is exactly of the type for which the Self-Avoiding McDiarmid Inequality gives upper bounds, provided we can
establish first that $h$ satisfies the required conditions of the Self-Avoiding McDiarmid inequality. Deriving such a bound is thus our first step.
\vskip-0.25cm
From (\ref{h_eq}) and Lemma 2.2 we get
\begin{eqnarray}
\nonumber&&h(\pi_1,\cdots,\pi_k)\,=\,\sum_{j=1}^k\,|\alpha_j\,|^2\,+\,\frac{1}{N}\,
\sum_{i,j \mbox{ with } i\neq j}\,\alpha_j\,\overline{\alpha_i}\,
\varphi_{\pi_j}^{\top}\,\overline{\varphi_{\pi_i}}\,.
\end{eqnarray}
We have then
\begin{eqnarray}
&&\!\!\!\!\!\!\!\!\!h(\pi_1,\,\ldots\,,\pi_{\ell},\,\ldots,\,\pi_k)\,-\,h(\pi_1,\,\ldots\,,\pi'_{\ell},\,\ldots\,,\,\pi_k)\nonumber\\
&&\!\!\!\hskip-1cm=\,\frac{1}{N}\,\sum_{j \mbox{ with } j\neq \ell}\,\left[\,\alpha_{\ell}\,\overline{\alpha_{j}}\,[\varphi_{\pi_{\ell}}\,-\,\varphi_{\pi'_{\ell}}]^{\top}\,\overline{\varphi_{\pi_j}}\,\right]
\nonumber\\&&\,\!\!\!\hskip-1cm+\,\frac{1}{N}\,\sum_{j \mbox{ with } j\neq \ell}\,\left[\alpha_{j}\,\overline{\alpha_{\ell}}\,\varphi_{\pi_j}^{\top}\,\overline{[\varphi_{\pi_{\ell}}\,-\,\varphi_{\pi'_{\ell}}]}\,\right]\nonumber\\
&&\hskip-1cm=\,\frac{1}{N}\,\sum_{}\,\left[\,\alpha_{\ell}\,\overline{\alpha_{j}}\,\sum_x\,\left(\,\varphi_{m(\pi_{\ell},\pi_j)}(x)\,-\,\varphi_{m(\pi'_{\ell},\pi_j)}(x)\,\right)\right]\,\nonumber\\
&&\hskip-1cm+\,\frac{1}{N}\,\sum_{}\,\left[\alpha_{j}\,\overline{\alpha_{\ell}}\,\sum_x\,\left(\,\varphi_{m(\pi_j,\pi_{\ell})}(x)\,-\,\varphi_{\pi(t_j,\pi'_{\ell})}(x) \,\right)\,\right]\,,\nonumber
%\label{eq11}
\end{eqnarray}
where we have used the same notation as in the proof of Lemma \ref{lemma3_2}, i.e.
$\varphi_{m(i,j)}(x):=\varphi_{i}(x)\,\overline{\varphi_{{j}}}(x)$. \\
Because $(\pi_1,\,\ldots\,,\pi_{\ell},\,\ldots,\,\pi_k)$ and $(\pi_1,\,\ldots\,,\pi'_{\ell},\,\ldots,\,\pi_k)$ are both in $\pk$, the indices $\pi_1,\,\ldots\,,\pi_{\ell},\,\ldots,\,\pi_k$ and $\pi'_{\ell}$ are all different. It then follows from
(\ref{st3}) that
\begin{eqnarray}
&&\hskip-0.5cm\!\!\!\left|\,h(\pi_1,\,\ldots\,,\pi_{\ell},\,\ldots,\,\pi_k)\right.-\left.h(\pi_1,\,\ldots\,,\pi'_{\ell},\,\ldots\,,\,\pi_k)\right|\quad\quad\,\nonumber\\
&&\hskip-1.5cm\quad\quad\leq\,\frac{2}{\mm}\,|\as_{\ell}| \,\sum_{}\,|\as_j|
\,\left|\,\sum_x\,\varphi_{\pi_{m({\ell},j)}}(x)\,-\,\varphi_{\pi_{m({\ell'},j)}}(x)\,\right|\nonumber\\
&&\hskip-1.5cm\quad\quad\leq\,\frac{2}{\mm}\,|\as_{\ell}| \,\sum_{j \mbox{\footnotesize{ with }} j\neq \ell}\,|\as_j|\, 2 \,\mm^{1-\eta/2}\,
\nonumber\\&=&\,
\frac{4}{\mm^{\eta/2}}\,|\as_{\ell}| \,\sum_{j \mbox{\footnotesize{ with }} j\neq \ell}\,|\as_j|
\label{eq12}
\end{eqnarray}
where we have used that $m(\pi_{\ell},\pi_j)\,\neq\,1$ if $\pi_{\ell}\,\neq\,\pi_j\,$, i.e.
if $\ell\neq j$.
Because this bound is uniform over the $\pi_{\ell}$ in $\cx_{\ell}$, it is now clear that this implies the sufficient condition of the Self-Avoiding McDiarmid inequality,
with $c_{\ell}$ given by the right hand side of (\ref{eq12}). We can thus conclude from 
Theorem \ref{McDthm} that 
\begin{eqnarray}\label{yess}
&&\hskip-1cm\Pr{\!}_{\pi}\left[|h - \Ex[h]| \geq \beta\,\|\as\|^2\,\right]\\\nonumber&&\hskip-1cm\leq 2\,\exp\left(-\,\frac{2\beta^2\,\mm^{\eta}\,\|\as\|^4}
{16\,\sum_{\ell=1}^k \,|\as_{\ell}|^2\,\left[\,\sum_{j \mbox{\footnotesize{ with }} j\neq \ell}|\as_j|\right]^2}\,\right).
\end{eqnarray}
Since
\begin{eqnarray}\nonumber
&&\sum_{\ell=1}^k \,
|\alpha_{\ell}|^2\,
\left[\,
\sum_{j \mbox{\footnotesize{ with }} j\neq \ell}
|\alpha_j|\right]^2 
\\\nonumber&\leq& \sum_{\ell=1}^k \,
|\alpha_{\ell}|^2\,
\left[\,\sum_{j=1}^k \,|\alpha_j|\right]^2 \\\nonumber&\leq& \|\alpha\|^2 \, k \, \|\alpha\|^2 = k \|\alpha\|^4~,
\end{eqnarray}
where we have used the Cauchy-Schwarz inequality in the penultimate step, it follows that
\[
\Pr{\!}_{\pi}\left[|h - \Ex[h]| \geq \beta\,\|\as\|^2\,\right]\leq 2\,\exp\left(-\,\frac{\beta^2\,\mm^{\eta}}{8k}\right)~.
\]
After substituting $\|\ff\|^2$ for $h$, and applying Lemma \ref{lemma3_2}, we finally obtain
that
\begin{eqnarray}\nonumber&&
\Pr{\!}_{\pi}\left[|\|\ff\|^2 - \|\as\|^2| \geq \left(\beta+\frac{k-1}{\n-1}\right)\|\as\|^2\,\right]\\\nonumber&\leq& 2\,\exp\left(-\,\frac{\beta^2\,\mm^{\eta}}{8k}\right)~.
\end{eqnarray}
For $\epsilon > (k-1)/(\n-1)$, we can set $\beta=\epsilon -(k-1)/(\n-1)$, thus recovering the StRIP-bound claimed in the statement of Theorem \ref{mainthm} for this case:
}}
with probability at least $1-2\exp\left(\frac{-\left[\epsilon-(k-1)(\n-1)^{-1}\right]^2 \mm^\eta}{8k}\right)$, we have the following near-isometry for $k$-sparse vectors $\as$:
\begin{equation}\label{distor}
(1-\epsilon)\|\as\|^2 \leq \|{f}\|^2\leq (1+\epsilon) \|\as\|^2.
\end{equation}
\end{proof}

\begin{rem}\label{firstfinish}
Equation~(\ref{distor}) implies that as long as 
$\sqrt{\frac{k}{\mm^\eta}} +\frac{k-1}{\n-1}
{\icd{<}} \epsilon\,$, the probability of failure 
{\icd{(i.e. the probability that the near-isometry inequality fails to hold)}} 
drops to zero 
{\icd{as $\n \rightarrow \infty$}}. 
{\icd{In particular,}} if $\eta$ equals 1, $k \leq \mu (\n-1)\epsilon+1$ for some constant $\mu$ less than one, and 
%\cb{
$\mm=O\left(\frac{k \log\n}{\epsilon^2}\right)$
%} 
then the probability of failure approaches zero at the rate $\n^{-1}$. 
\end{rem}
\begin{rem}
Figure~\ref{fig:exp} shows the distribution of condition numbers for the singular values of restrictions of the sensing matrix to sets of K columns. Two cases are considered; the Reed Muller matrices constructed in Section 2.2 and random Gaussian matrices of the same size. The figure suggests that the decay of $$\Pr\left[\left| ||{f}||^2- ||{\alpha}||^2 \right| \geq \epsilon ||{\alpha}||^2 \right]$$ is similar for both types of compressive sensing matrices.
 \begin{figure*}
\begin{center}
\includegraphics[bb=0 0 360 252,width=0.99\textwidth,height=360pt]{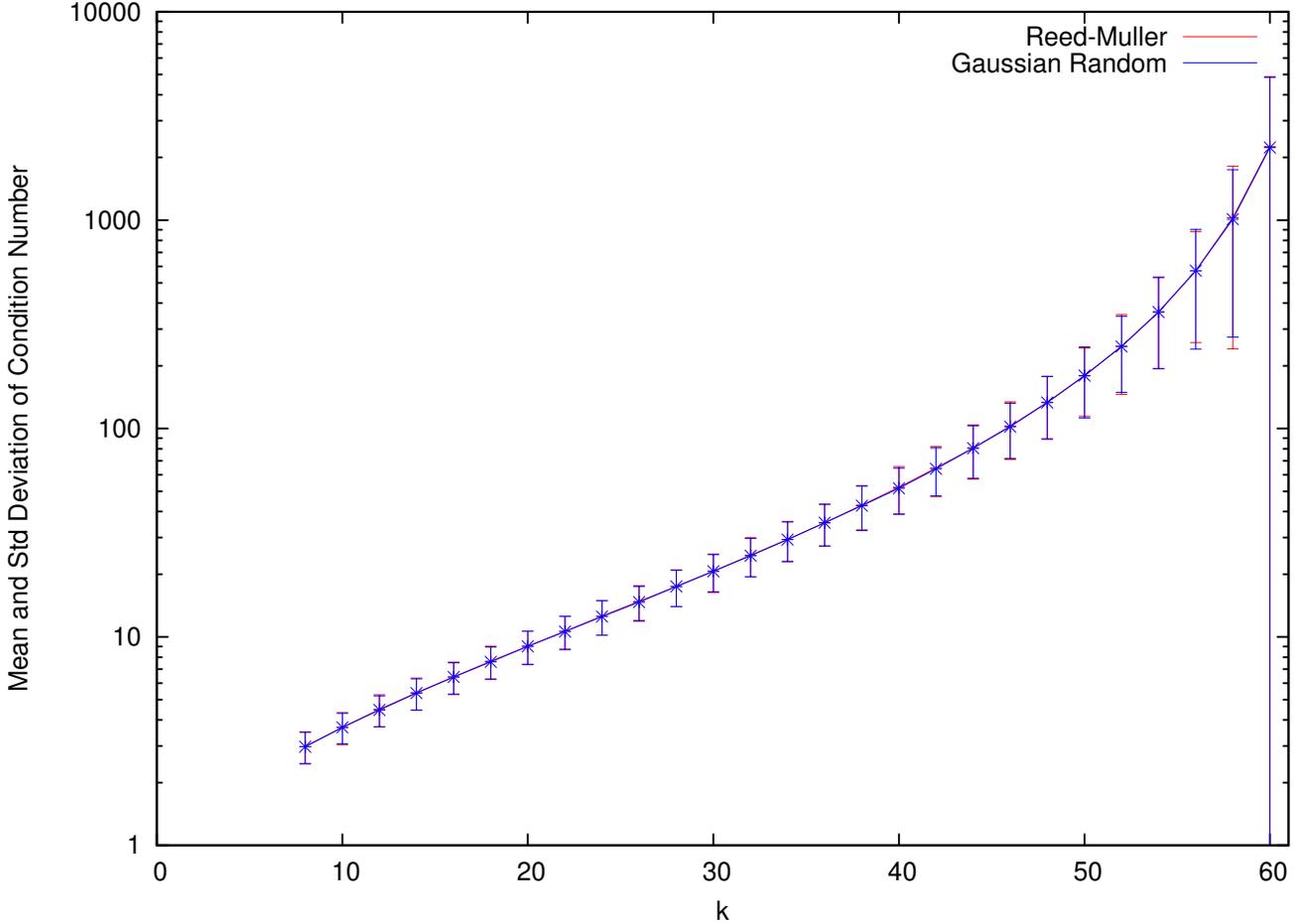}
\end{center}
\caption{Mean and standard deviation for the condition number of $k$-Gram matrices for $\Phi_{RM}$, with $m=6$, compared to
that of a Gaussian random matrix of the same size.}
\label{fig:exp}
\end{figure*}
\end{rem}

\begin{rem} Note that similar to {\icd{the case of}} random and expander matrices, the number of measurements $\mm$ {\icd{grows as the inverse square of}} the distortion parameter $\epsilon$, $\mm \propto \frac{1}{\epsilon^2}${\icd{, as $\epsilon\rightarrow 0$}}.
\end{rem}
\begin{rem}
By avoiding the use of the Cauchy-Schwarz inequality at the end of the proof, and making use 
of Remark \ref{rem_ell}, one can sharpen the bounds. From (\ref{yess}) it follows that with probability at least $ 1 - 2\exp\left(-\frac{\beta^2 \mm^\eta\|\as\|^2}{8\|\as\|^2_{\ell_1}}\right)$\vskip-0.4cm
$$
-\left(\beta -\frac{1}{\n-1}\right)\|\as\|^2 -
\frac{1}{\n-1}\,\|\as\|^2_{\ell_1}\leq\|\ff\|^2 - \|\as\|^2\,$$and
 $$
\|\ff\|^2 - \|\as\|^2\leq 
\left(\beta+\frac{1}{\n-1}\right)\|\as\|^2
$$
This implies (set $\beta = \gamma \rho$, $\rho=\|\as\|_{\ell_1} \|\as\|^{-1}\,$)
\begin{eqnarray}\nonumber&&\hskip-0.5cm
\Pr{\!}_{\pi}\left[\left|\|\ff\|^2 - \|\as\|^2 \right| \geq 
\left( \gamma\rho + \frac{1}{\n-1}\max(1,\rho^2-1)\right)\|\as\|^2 
\right]\\\nonumber&&\hskip-0.5cm\leq 2\,e^{-\,\gamma^2\, \mm^\eta/8}~,
\end{eqnarray}
or, equivalently,
\begin{eqnarray}
\nonumber&\,&\Pr{\!}_{\pi}\left[\left|\|\ff\|^2 - \|\as\|^2 \right| \geq \epsilon \|\as\|^2 
\right]\\\nonumber &\leq& 2\,\exp \left(- \frac{1}{8} \mm^\eta \left[\frac{\epsilon -(\n-1)^{-1}}{\rho}-\chi_{(\rho>\sqrt{2})}\frac{\rho^2-2}{\rho(\n-1)}\right]^2\right)~,
\end{eqnarray}
with $\rho=\|\as\|_{\ell_1} \|\as\|^{-1}\,$, as above, and $\chi_{(a>0)}=1$ if $a> 0$,
$\chi_{(a>0)}=0$ otherwise. The worst case for this bound
is when $\rho=\sqrt{k}$, in which case we recover the bound in Theorem \ref{mainthm}; if one is restricted, for whatever reason, to $k$-sparse vectors that are known to have some 
entries that are much larger than other non vanishing entries, then the more complicated 
bound given here is tighter.
\end{rem}
% \cb

{
\begin{rem} \label{r1misha}If the sparsity level $k$ is greater than $\sqrt{\n}$, then $\n\leq k^2\leq \mm^2$. However, since some deterministic sensing matrices of section~\ref{sec:matrices} structurally require the condition $\mm^2 \leq \n$, a deterministic matrix with $\mm'=O\left( \frac{k\log \n}{\epsilon^2}\right)$ rows and $\mm'^2\leq\n'$ columns is required. 
In this case, the $\mm'\times\n$ sensing matrix $\A$ is constructed by choosing $\n$ random columns from the $\mm'\times \n'$ deterministic matrix. 
\end{rem}
%}
}
}

\subsection{Proving UStRIP: Uniqueness of Sparse Representation}

Although we have established the desired near-isometry bounds, we still have to address the Uniqueness guarantee; unlike the standard RIP case, this does not follow automatically from a StRIP bound, as pointed out in the Introduction. More precisely, we need to estimate the probability that a randomly picked $k$-sparse vector $\as$ has an ``evil twin''
$\as'\neq \as$ that maps to the same image under $\A$, i.e. 
$\A \as=\A \as' $, and prove that this probability is very small. 

If $S\subset \set$ is the 
union of possible support sets 
of a two 
$k$-sparse vectors, that is, if
$s\doteq |S|\leq 2k$, then we define $\A _S$ to be the $N\times s$ matrix obtained by 
picking out only the columns indexed by labels in $S$. In other words, the matrix elements of
$\A _S$ are those $\varphi_j(x)$ for which $j \in S$, with $x$ varying over its full range.
There will be two different $k$-sparse vectors $\as'\neq \as$, the supports of which are both contained in $S$, if and only if the $s \times s$ matrix $\A^{\dag}\A$ is rank-deficient (where $\A^\dag$ denotes conjugate transpose of $\A$). Note that this property concerns the support set $S$ only -- the values of the entries of $\as$ are not important. This is similar to the discussion of sparse reconstruction when $\A$ satisfies a deterministic Null Space Property \cite{Devore}.
Once uniqueness is found to be overwhelmingly likely, we can derive from it the probability that decoding algorithms (such as the quadratic decoding algorithms described in 
Section \ref{sec:quad})
succeed in constructing, from $\A \as$, a faithfully exact or close copy (depending on the application) of the $k$-sparse source vector $\as$.

In fact, it turns out that we won't even have to consider matrices $\A_S$ with 
$|S|=2k$; as we shall see below, it suffices to consider $\A_S$ for sets $S$ of 
cardinality up to $k$.

Once again, condition (St3) will play a crucial role. For the StRIP analysis, in the previous subsection, it sufficed to to take $\eta >0$, where $\eta$ is the
parameter that measures the closeness of column sums in (St3).
In this subsection, we will impose a non-zero lower bound on $\eta$; we shall see that $\eta{\icd{>}}0.5$  suffices for our analysis. 

We recall here the formulation of (St3): for any column $\varphi_j $ 
of the sensing matrix, with  $j \geq 2$,
\[
\left|  \sum_x  \varphi_j(x) \right| \leq \mm^{1-\eta/2} .
\] 
We 
introduced the notation 
$\A_S$ 
at the start of this subsection. We shall
also use the special case where we wish to restrict
the sensing matrix $\A$ 
to a single column indexed by $w$; in that case, we denote 
the restriction by $\varphi_w$. Finally we denote the conjugate transpose of a matrix $\A_S$ by $\A_S^\dag$. We shall use Tropp's argument (see Section 7 of [Tro08b]) to prove uniqueness of sparse representation; to apply this argument we first 
need to prove that a random submatrix $\varphi_\kappa$ has small coherence with the remaining columns of the sensing matrix.
\begin{lemma}
\label{ul1}
Let $\Phi$ be $\eta-$StRIP-able with $\eta >1/2$, and assume that the conditions $k< \epsilon(C-1)+1$, and
$N=O\left(\left(\nicefrac{k \log\n}{\epsilon^2}\right)^{1/\eta}\right)$ hold, and $\delta$ is as defined in Theorem \ref{maintheorem}, i.e $$\delta\,:=\,2\exp\left[\,-\,\frac{{\icd{[\epsilon-(k-1)/(\n-1)]}}^2\,\mm^{\eta}}{8\, k}\,\right].$$ Let $w$ be a fixed column of $\A$, and let $\kappa=\left\{ \kappa_1,\cdots,\kappa_{k}\right\}$ be the positions of the first $k$ elements of a random permutation of 
{\cb{ 
$\set \setminus \{w \} .$
}}
Then
\begin{equation}
\label{ue1}
\Ex\left[\left\|\frac{1}{\sqrt{\mm}}\A^{\dag}_\kappa \frac{1}{\sqrt{\mm}}\varphi_w\right\|^2\right]=
\frac{k}{\mm}\frac{\n-\mm}{(\n-1)},
\end{equation}
where the expectation is with respect to the choice of the set $\kappa$.
\end{lemma}
\begin{proof}
By linearity of expectation we have
\begin{equation}\label{ue2}
\Ex_\kappa\left[\left\|\frac{1}{\sqrt{\mm}}\A^{\dag}_\kappa \frac{1}{\sqrt{\mm}}\varphi_w\right\|^2\right]=\frac{1}{\mm^2}\sum_{i=1}^{k}
\Ex_{\kappa_i}\left[\left|\left(\overline{\varphi_{\kappa_i}}\right)^\top \varphi_w\right|^2\right].
\end{equation}
Since the set of columns of $\A$  is invariant under complex conjugation, and forms a group under pointwise multiplication, we have
\begin{eqnarray}
\nonumber \left(\overline{\varphi_{\kappa_i}}\right)^\top \varphi_w&=&\sum_{x} \overline{\varphi_{\kappa_i}(x)}\varphi_w(x)= \sum_{x} \varphi_{m(w,\kappa_i)}(x),
\end{eqnarray}
where we use again the notation introduced just below (\ref{eqn6}): 
$\varphi_{\ell}(x)\overline{\varphi_{\ell'}(x)}\doteq \varphi_{m(\ell,\ell')}(x)$.
As $\kappa$ ranges over all the possible permutations that do not move $w$, $\kappa_i$ ranges uniformly 
over $\set\setminus\{w\}$, and the different $z_i:=m(w,\kappa_i)$ range uniformly
over $\{2,\ldots,\n\}$.

Hence:
{\cb{
\begin{eqnarray}
\nonumber&&\sum_{i=1}^{k}
\frac{1}{\mm^2}
\Ex_{\kappa_i}\left[\left|\left(\overline{\varphi_{\kappa_i}}\right)^\top \varphi_w\right|^2\right]
\\\nonumber&=&\sum_{i=1}^{k}
\frac{1}{\mm^2} 
\Ex_{z_i}\left[\left|\sum_x\varphi_{z_i}(x)\right|^2\right]
\\\nonumber&=&\sum_{i=1}^k\frac{1}{\mm^2} 
\Ex_{z_i}\left[\sum_{x,y=1}^N \varphi_{z_i}(x)\overline{\varphi_{z_i}(y)}\right]
\nonumber\\
&=&
\frac{k}{\mm^2}
\frac{1}{(\n-1)} \sum_{j=2}^{\n} \sum_{x,y=1}^N \varphi_j(x)\overline{\varphi_j(y)}\\\nonumber&=&
\frac{k}{\mm^2}
\frac{1}{(\n-1)} \left( \sum_{x,y=1}^N \left[\n\delta_{x,y}-1\right] \right)\nonumber\\
&=& 
\frac{k}{\mm^2} \frac{1}{(\n-1)}(\mm\n-\mm^2)= \frac{k}{\mm}\frac{\n-\mm}{(\n-1)},\nonumber
\end{eqnarray}
}}
where we have used (St1) .
\end{proof}

Next, we use the Self-Avoiding McDiarmid  inequality, together with property~(St3) to derive a uniform bound for the random variable $\left\|\A_\kappa^\dag \varphi_w\right\|^2$:

\begin{theorem}
\label{t1}
Let $\Phi$ be $\eta-$StRIP-able with $\eta >1/2$, and assume that the conditions $k< \epsilon(C-1)+1$, and
$N=O\left(\left(\frac{k \log\n}{\epsilon^2}\right)^{1/\eta}\right)$ hold, define $\delta$ as in Theorem~\ref{maintheorem}, and let 
$\lambda$ be a set of $k$ random columns of $\A$. 
Then with probability at least $1-\delta$, there exists no $w$ such that
\begin{equation}
\label{ue5}
\left\|\frac{1}{\sqrt{\mm}}\A_\lambda^\dag \frac{1}{\sqrt{\mm}}\varphi_w\right\|^2
\geq \frac{k}{\mm}\frac{\n-\mm}{\n-1}+
\frac{\sqrt{2k\log \nicefrac{\n}{\delta}}}{\mm^\eta}
\end{equation}
\end{theorem}
\begin{proof}
{\cb{
The proof is in several steps. In the first step, we pick any $w \in \set$, and keep it fixed (for the time being). Let $$f(t_1,\cdots,t_k)=
\frac{1}{\mm^2}
\sum_{i=1}^{k} \left|\left(\overline{\varphi_{t_i}}\right)^\top \varphi_w \right|^2,$$ 
where we assume that $t_1, \cdots, t_k$ are k different elements of $\set\setminus\{w\}$, picked at random. Note that if $\lambda$ is a random permutation of $\set\setminus\{w\}$,
then $f(\lambda_1,\ldots,\lambda_k)=\left\|\frac{1}{\sqrt{\mm}}
\A_\lambda^\dag 
\frac{1}{\sqrt{\mm}}
\varphi_w\right\|^2$The function $f$, as defined above  from 
$$\{(t_1,t_2,\ldots,t_k)\,; \,t_i\in \set\setminus\{w\}\mbox{ $\forall$ }i, \, t_i\neq t_j,
\mbox{ $\forall$ } i\neq j\}$$ to $\RR$, is information-theoretically indistinguishable from the function $F$ from
the  permutations of  $\set\setminus\{w\}$ to $\RR$ defined by $$F(\lambda)=\left\|\frac{1}{\sqrt{\mm}}
\A_\lambda^\dag 
\frac{1}{\sqrt{\mm}}
\varphi_w\right\|^2.$$
We have computed $\Ex[f]=\Ex[F]$
in Lemma \ref{ul1}; in order to apply the Self-Avoiding McDiarmid Inequality to $f$,
we need verify only that a necessary condition of the Self-Avoiding McDiarmid inequality
holds. 
}}
\\When we subtract $f(t_1,\cdots,t_{i-1},t'_{i},t_{i+1},\cdots,t_{k})$ from 
$f(t_1,\cdots,t_{i-1},t_i,t_{i+1},\cdots,t_{k})$, only the $i$-th term survives; we have 
\begin{eqnarray}
\nonumber&&\Ex[f(t_1,\cdots,t_{i-1},t_i,t_{i+1},\cdots,t_{k})]\\\nonumber&-&\Ex[f(t_1,\cdots,t_{i-1},t'_{i},t_{i+1},\cdots,t_{k})\\
&=&\left|
\frac{1}{\mm^2} 
\left|\left(\overline{\varphi_{t_i}}\right)^\top \varphi_w \right|^2
-\frac{1}{\mm^2} 
\left|\left(\overline{\varphi_{t'_i}}\right)^\top \varphi_w \right|^2\right|
\nonumber\\
&=& \frac{1}{\mm^2}
\left|\left|\sum_{x=1}^\mm\varphi_{m(w,t_i)}(x)\right|^2 - 
\left|\sum_{x=1}^\mm\varphi_{m(w,t'_i)}(x)\right|^2\right|\nonumber\\
&\leq& 2\mm^{-\eta}~,
\end{eqnarray}
by (St3), 
since $m(w,t_i)\neq 1\neq m(w,t'_i)$. It immediately follows that the concentration condition holds for $f$, with 
$c_i = \mm^{-\eta}$. Therefore the Self-Avoiding McDiarmid Inequality holds for $f$, which means it also holds for $F$:
for any positive $\gamma$,
\begin{eqnarray} 
&&\Pr{\!}_\lambda\left[\left\|\frac{1}{\sqrt{\mm}}\A_\lambda^\dag \frac{1}{\sqrt{\mm}}
\varphi_w\right\|^2\geq \frac{k}{\mm}+\gamma\right] \nonumber\\
&&\hskip-0.7cm\quad\quad\leq\Pr{\!}_\lambda\left[\left\|\frac{1}{\sqrt{\mm}}\A_\lambda^\dag
\frac{1}{\sqrt{\mm}}\varphi_w\right\|^2\geq \frac{k}{\mm}\frac{\n-\mm}{\n-1}+\gamma\right]\nonumber \\\nonumber&&\leq \exp\left(\frac{-\gamma^2\mm^{2\eta}}{2k}\right).\nonumber
\end{eqnarray}
All this was for one fixed choice of $w$; note that the bound does not depend on the
identity of $w$.
This implies that by applying union bounds over the $\n$ possible choices for the column
$w$ of $\A$, we get that the probability that there exists a $w$ such that 
\[
\left\|\frac{1}{\sqrt{\mm}}\A_\lambda^\dag \frac{1}{\sqrt{\mm}}\varphi_w\right\|^2\geq 
\frac{k}{\mm}+\gamma, \]
is at most $ \n\exp\left(\frac{-\gamma^2\mm^{2\eta}}{2k}\right).
$ Writing $\gamma$ in terms of $\delta$ completes the proof. 
\end{proof}

If $\mm=O\left(\left( \frac{k \log \n}{\epsilon^2}\right)^{1/\eta} \right)$, 
the right hand side of (\ref{ue5}) reduces to
\[
O \left(\frac{\epsilon^2}{N}\right) + 
O \left (  \epsilon^{2\eta} \left[1+\frac{|\log \delta|}{\log \n}\right]^{1/2}(k\log \n)^{-(\eta-1/2)}    \right)
\]
Thus, if $\eta > 1/2$, then (for sufficiently small $\epsilon$, and sufficiently large $\n$)
a choice of $k$ random columns of $\A$
has a very high probability of having small coherence with {\em any} other column of the matrix; in particular, 
{\cb{
we have, with probability exceeding $1-\delta$, that
\begin{equation}\label{rob}\left\|\frac{1}{\sqrt{\mm}}\A_\lambda^\dag \frac{1}{\sqrt{\mm}}\varphi_w \right\|^2 < (1-\epsilon)^2.\end{equation}
}}

This establishes incoherence between the random submatrix $\A_\lambda$ and the remaining columns of the sensing matrix. 

We can now complete the UStRIP proof by following an argument
of Tropp \cite{joel1}; for completeness we include the argument here:
\begin{lemma}
\label{mahnaz} Let $\lambda=\{\lambda_1,\cdots,\lambda_{k}\}$ be a set of $k$ indices sampled uniformly from $\yc$. Assume that $\A$ {\icd{is}} $(k,\epsilon{\icd{,\delta}})$-StRIP. Let $S$ be any other subset of $\yc$ of size less than or equal to $k$. Then{\icd{, 
with probability at least $(1-\delta)$ (with respect to the randomness in the choice 
of $\lambda$) }}
\begin{equation}
\label{maheq}
\dim\left(\range(\A_\lambda)\cap \range(\A_S)\right)<k.
\end{equation}
\end{lemma}
\begin{proof} First, note that we {\icd{need check only the case}}  $\dim\left(\range(\A_S)\right)=k$, {\icd{since}} otherwise ~(\ref{maheq}) {\icd{is immediate}}. 
{\icd{Note also that, because $\A$ is $(k,\epsilon,\delta)$-StRIP, the probability that
the randomly picked set 
$\lambda=\{\lambda_1,\cdots,\lambda_{k}\}$ satisfies
\[
(1-\epsilon)\mbox{Id}_{\lambda}\leq \frac{1}{N}\A_{\lambda}^{\dag} \A_{\lambda} \leq (1+\epsilon)\mbox{Id}_{\lambda}
\]
is at least $1-\delta$. (The notation $\mbox{Id}_{\lambda}$ stands for the identity matrix on $\lambda$;  this  just amounts to restating the 
$(k,\epsilon,\delta)$-StRIP condition in matrix form.) It follows that, with probability at least $1-\delta$,}}
\begin{equation}
\label{StRIPdelta}
\sigma_{\min}\left(\A_\lambda\right)\geq {\icd{\sqrt{(1-\epsilon)\mm}}},
\end{equation} where $\sigma_{\min}   \left(\A_\lambda\right) $ 
is the smallest singular value of $\A_\lambda$. \\
Since $S\neq \lambda$, $S$ has at least one index not in $\lambda$. Denote that index by $s$. Since the entries of the matrix are all unimodular{\icd{,}} we have
\begin{equation}
\label{ueq7}
\left\|\varphi_s\right\|^2=\sum_x \left|\varphi_s(x)\right|^2=\mm.
\end{equation}
Let $\PT$ be the orthogonal projection operator on the {\icd{range $\mathcal{R}_{\lambda}$
of
$\A_\lambda$}}. We {\icd{shall}} prove~(\ref{maheq}) by showing that $\|\PT\varphi_s\|^2<\|\varphi_s\|^2$, which implies that there exists a vector in the range of $\A_S$ that is outside the range of $\A_\lambda$. Note that 
\begin{equation}
\label{ueq8}
\PT={\icd{\A_\lambda\left(\A_\lambda^\dag \A_\lambda\right)^{-1}}} \A_\lambda^\dag{\icd{.}}
\end{equation}
Since $\A_\lambda$ {\icd{is}} $(k,\epsilon{\icd{,\delta}})-$StRIP, {\icd{we have, still with probability at least $1-\delta$,}}
\begin{eqnarray}
\nonumber
\left\|\PT \varphi_s\right\|^2 
&=&
{\icd{ 
\left(\A_{\lambda}^{\dag}\varphi_s\right)^{\dag}
\left(\A_\lambda^\dag \A_\lambda\right)^{-1}
\left(\A_{\lambda}^{\dag}\varphi_s\right)
}}
\\\nonumber
&\leq& \frac{\left\|\A_{\lambda}^\dag \varphi_s\right\|^2}{\left(\sigma_{\min}(\A_\lambda)\right)^2}
\leq \frac{\left\|\A_{\lambda}^\dag \varphi_s\right\|^2 }
{\mm (1-{\icd{\epsilon }}) }
\\\nonumber&\leq& (1-\epsilon) \mm
< \mm,
\end{eqnarray}
where the {\icd{penultimate}} inequality is by Equation~(\ref{rob}).
\end{proof}

\begin{theorem}
\label{ueq6}
Let $\Phi$ be $\eta-$StRIP-able with $\eta >1/2$, and assume that the conditions $k< \epsilon(C-1)+1$, and
$N=O\left(\left(\frac{k \log\n}{\epsilon^2}\right)^{1/\eta}\right)$ hold, define $\delta$ as in Theorem~\ref{maintheorem}, 
and let $\as$ be {\icd{a randomly picked}} $k$-sparse signal.
{\icd{Then with probability at least $1-\delta$ (with respect to the random choice of $\as$)}}, $\as$ is the only $k$-sparse vector that satisfies the equation ${f}=\nor\A\as$.
\end{theorem}
\begin{proof}
{\cb{
We have already proved in Section 3.1.2 that $\Phi$ is $(k,\epsilon,\delta)$-StRIP. We start by recalling that the random choice of $\as$ can be viewed as first choosing
its support, a uniformly distributed subset of size $k$ within $\{1,\cdots,\n\}$, and then,
once the support is fixed, choosing a random vector within the corresponding
$k$-dimensional vector space. For this last choice no distribution has been specified;
we shall just assume that it is absolutely continuous with respect to the Lebesgue
measure on $\RR^k$ or $\CC^k$. \\
Since 
$\A$ is $(k,\epsilon,\delta)$-StRIP, $\A_{\lambda}$ is non-singular with
probability exceeding $(1-\delta)$, so that
$$\dim\left(\range(\A_\lambda)\right)=k$$ with probability exceeding 
$1-\delta$. The near-isometry property of $\A_\lambda$ implies that no two signals with support 
$\lambda$ can have the same value in the measurement domain. 
If there nevertheless were a vector $\as'$ such that $\A \as'=\A\as$, the support $S$
of $\as'$ would 
therefore necessarily be different from $\lambda$. By Lemma~\ref{mahnaz}, we know that 
$V \doteq \mbox{range} (A_\lambda) \cap \mbox{range}(A_S)$
is at most $(k-1)$-dimensional.
It follows that in order to possibly have an ``evil twin'' $\as'$, the
vector $\as$ must itself lie in the at most $(k-1)$-dimensional
space
that is the inverse image of $V$ under $\A_\lambda$. This set, however, has measure
zero with respect to  any measure that is absolutely continuous with respect
to the $k$-dimensional Lebesgue measure. Thus, for each $k$-set $\lambda$ for which $\A_\lambda$ is a near-isometry, 
the vectors that are not uniquely determined by their image $\A\as$, constitute
a set of measure zero. Since randomly chosen $k$-sets $\lambda$ produce restrictions
$\A_\lambda$ that are near-isometric with probability exceeding $1-\delta$, the
theorem is proved.
}}
\end{proof}
Combining Remark~\ref{firstfinish} with Theorem~\ref{ueq6} completes the proof of Theorem~\ref{maintheorem}.

\label{sec:null}
\section{Partial Fourier Ensembles}
\label{sec:fourier}
\cb{
 In Partial Fourier ensembles the matrix $\A$ is formed by uniform random selection of $\mm$ rows from the $\n\times\n$ discrete Fourier Transform matrix. The resulting random sensing matrices are widely used in compressed sensing, because the corresponding memory cost is only $O(\mm\log \n )$, in contrast to the $O(\mm \n)$ cost of storing Gaussian and Bernoulli matrices. Moreover, it is known \cite{CT,NT} that if $\mm \geq k \log^5 \n$, then with overwhelming probability, the partial Fourier matrix satisfies the RIP property. It is easy to verify that  such$\A$ satisfies the Conditions~(St1), and~(St2). We now show that it also satisfies Condition~(St3) almost surely. 

Note that here in contrast to the proof of Theorem~\ref{maintheorem}, the randomness is with respect to the choice of the $\mm$ rows from the Discrete Fourier Transform matrix. We show that with overwhelming probability, the condition~(St3) is satisfied for every column of this randomly sampled matrix. First fix a column $\varphi_i$ other than the identity column, and define the random variable $Z_x$ to be the value of the entry $\varphi_i(x)$, where the randomness is with respect to the choice of the rows of $\A$ (that is with respect to the choice of $x$). Since the rows are chosen uniformly at random, and the column sums (for all but the first column) in the discrete Fourier transform are zero, we have
\begin{equation}
\label{rev2}
\mathbb{E}\left[\frac{\sum_x Z_x}{\mm}\right]=\frac{\sum_x \mathbb{E}[Z_x] }{\mm}=0.
\end{equation}
Since all entries are unimodular, we may apply Hoeffding's inequality to both the real and the imaginary part of the random variable $\frac{\sum_x Z_x}{\mm}$, then apply union bounds to conclude that for all $\epsilon > 0$
$$
\Pr\left[\left| \frac{\sum_x Z_x}{\mm} \geq \epsilon  \right| \right]\leq 4\exp\left\{ -2\mm\epsilon^2\right\} .
$$ 
Applying union bounds to all $\n-1$ admissible columns we get
\begin{equation}
\label{rev3}
\Pr\left[\mbox{there exists a column average greater than } \epsilon\right]
\end{equation}
is at most $ 4\n\exp\left\{ -2\mm\epsilon^2\right\} .$
Hence, with probability at least $1-\delta$ all column averages are $O\left(\sqrt{ \frac{\log\n}{\mm}}\right)$, and all column sums are less than $\sqrt{\mm\log \n}$, so that condition~(St3) is indeed satisfied. Applying Theorem~\ref{maintheorem} we see that a partial Fourier matrix satisfies StRIP with only $k ~\log\n$ measurements. This improves upon the best previous upper bound of $k \log^5 \n$ obtained in \cite{CT} and helps explain why partial Fourier matrices work well in practice.
}
\section{Quadratic Reconstruction Algorithm}
\cb{
\begin{algorithm}[ht]
\nonumber
\caption{ Quadratic Reconstruction Algorithm}
   Input: $\mm$ dimensional vector ${f}=\sqm \A \bs{\alpha}+\nu$\\
   Output: An approximation $\has$ to the signal $\bs{\alpha}$\\
     \label{alg1}
   \begin{algorithmic}[1]
   \STATE Set $\yy_1=\yy$, $\Theta=\{\}$, $\has=0_\mm$.
   \FOR{$t=1,\cdots,k$ or while $\|{f_t}\|_2 \geq \epsilon$}
   \FOR{ each entry $x=1$ to $\mm$}\label{a1}
   \STATE \label{mul} pointwise multiply $\yy_t$ with a shift (offset) of itself as in~(\ref{walsh}).
   \ENDFOR 
   \STATE Compute the fast Walsh-Hadamard transform of the pointwise product: Equation~(\ref{fourier}) \label{a2}
   \STATE Find the position $p_t$ of the next peak in the Hadamard domain: Equation~(\ref{chirp}) implies that the chirp-like cross terms appear as a constant background signal. \label{a3}
   \IF{$p_t \in \mbox{Keys}(\Theta)$} \label{a4}
   \STATE Restore ${f_t}\leftarrow {f_t}+\Theta(p_t) \varphi_{{p_t}}$.
   \ENDIF
   \STATE  \label{opt2} Update  ${\beta_t}\doteq\sqm \yy^\top \varphi_{p_t} $ which minimizes $\|{f_t}-\sqm\beta_t \varphi_{p_t}\|_2$.  
   \STATE Add $\beta^t$ to entry $p_t$ of $\has$.
   \STATE Set $\Theta(p_t)= \beta_t.$
   \STATE Set ${f_{t+1}}\leftarrow {f_t}- \beta_t \varphi_{p_t}$. \label{a5}
   \ENDFOR
   \end{algorithmic}
   \end{algorithm}
}
The Quadratic Reconstruction Algorithm  \cite{LHSC,HSC,quad}, described in detail above, takes advantage of the multivariable quadratic functions that appear as exponents in Delsarte-Goethals sensing matrices. It is this structure that enables the algorithm to avoid the matrix-vector multiplication required when Basis and Matching Pursuit algorithms are applied to random sensing matrices. Because our algorithm requires only vector-vector multiplication in the measurement domain, the reconstruction complexity is sublinear in the dimension of the data domain. \cb{The Delsarte-Goethals sensing matrix was introduced in Section 2.2: there are $2^m$ rows indexed by binary $m$-tuples $x$, and $2^{(r+2)m}$ columns $\varphi_{P_i,b_i}$ indexed by pairs $P_i, b_i$ where $P_i$ is a binary symmetric matrix and $b_i$ is a binary $m$-tuple. The first step in our algorithm is pointwise multiplication of a sparse superposition 
$$f(x)=\nor \sum_{i=1}^k \alpha_i\varphi_{P_i,b_i}(x)$$
 with a shifted copy of itself. The sensing matrix is obtained by exponentiating multivariable quadratic functions so the first step produces a sparse superposition of pure frequencies (in the example below, these are Walsh functions in the binary domain) against a background of chirp-like cross terms.
\begin{eqnarray}
\label{walsh}
&~&f(x+a)\overline{f(x)}=\frac{1}{\mm}\sum_{j=1}^k |\alpha_j|^2(-1)^{a^\top P_jx}\\\nonumber&+&\frac{1}{\mm} \sum_{j\neq t} \alpha_j \overline{\alpha_t} \varphi_{P_j,b_j}(x+a) \overline{\varphi_{P_t,b_t}(x)}.
\end{eqnarray}

Then the (fast) Hadamard transform concentrates the energy of the first term $\frac{1}{\mm}\sum_{j=1}^k |\alpha_j|^2(-1)^{a^\top P_jx}$ at (no more than) $k$ Walsh-Hadamard tones, while the second term distributes energy uniformly across all $\mm$ tones. The $l^{th}$ Fourier coefficient is
\begin{equation}
\label{fourier}
\Gamma_a^l=\frac{1}{\mm^{\nicefrac{3}{2}}} \sum_{j\neq t} \alpha_j \overline{\alpha_t} \sum_x (-1)^{l^\top x} \varphi_{P_j,b_j}(x+a) \overline{\varphi_{P_t,b_t}(x)},
\end{equation}
and it can be shown (see \cite{quad}) that the energy of the chirp-like cross terms is distributed uniformly in the Walsh-Hadamard domain. That is for any coefficient $l$
\begin{equation}
\label{chirp}
\mbox{lim}_{\mm\rightarrow \infty} \Ex\left[\mm^2 \left|\Gamma_a^l\right|^2\right]=\sum_{j \neq t} |\alpha_j|^2 |\alpha_t|^2.
\end{equation}
Equation~(\ref{chirp}) is related to the variance of ${f}$ and may be viewed as a fine-grained concentration estimate. In fact the proof of (\ref{chirp}) mirrors the proof of the UStRIP property given in Section 3; first we show that the expected value of any Walsh-Hadamard coefficient is zero, and then we use the Self-Avoiding McDiarmid Inequality to prove concentration about this expected value.  The Walsh-Hadamard tones appear as spikes above a constant background signal and the quadratic algorithm learns the terms in the sparse superposition by varying the offset $a$. These terms can be peeled off in decreasing order of signal strength or processed in a list. The quadratic algorithm is a repurposing of the chirp detection algorithm commonly used in navigation radars which is known to work extremely well in the presence of noise. Experimental results show close approach to the information theoretic lower bound on the required number of measurements. For example, numerical experiments show that the quadratic decoding algorithm is able to reconstruct greater than $40$-sparse superpositions when applied to deterministic Kerdock sensing matrices with $N = 2^9$ and ${\mathcal C} = 2^{18}$. In this case, the information theoretic lower bound is $k\log_2(1 + {\mathcal C}/k) = 507$ \cite{HSC}.
 \begin{figure}[ht]
\begin{center}
\includegraphics[bb= 0 0 504 416,width=0.5\textwidth,height=0.35\textwidth]{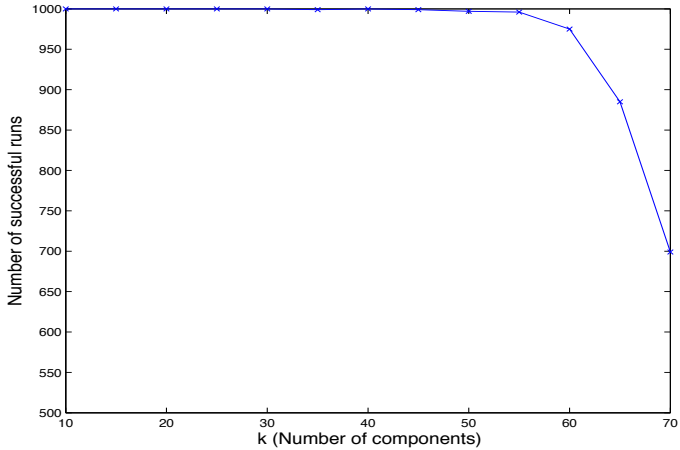}
\end{center}
\caption{ The number of successful reconstructions in $1000$ trials versus the sparsity factor $k$ for the deterministic Kerdock sensing matrix
corresponding to $m=9$}
\label{fig:info}
\end{figure}

We now explain how the StRIP property provides performance guarantees for the Quadratic Reconstruction Algorithm. At each iteration the algorithm returns the location of one of the $k$ significant entries and an estimate for the value of that entry. The StRIP property guarantees that the estimate is within $\epsilon$ of the true value. These errors compound as the algorithm iterates, but since the chirp cross-terms and noise are uniformly distributed in the Walsh-Hadamard domain, the error in recovery is bounded by the difference between the true signal  $\alpha$ and its best $k$-term approximation $\alpha_k$. More precisely, if $\A$ is $(k,\epsilon,\delta)$-StRIP, if the position of the $k$ significant entries are chosen uniformly at random, if the near-zero entries and the measurement noise $\nu$ come from a Gaussian distribution, and if the Quadratic Recovery Algorithm is used to recover an approximation $\has$ for $\as$, then
\begin{equation}
 \|\as-\has\|_2\leq \frac{5+\epsilon}{1-\epsilon} \|\as-\as_k\|_2+\frac{2}{1-\epsilon}\|\nu\|_2.\end{equation}
 The role of the StRIP property is to bound the error of approximation in Step 11 of the Quadratic Reconstruction Algorithm. Note that if it were somehow possible to identify the support of $\alpha$ beforehand, then the UStRIP property would guarantee that we would be able to recover the signal values by linear regression. However identifying the support of a $k$-sparse signal is known to be almost as hard as full reconstruction, and that is why our algorithm finds location and estimates signal value simultaneously, and does so one location at a time.

Note that the error bound is of the form $\ell_2/\ell_2$:
\begin{equation}\label{delta}
{\|\as-\has}\|_2 \leq C \|\as-\as_k\|_2.
\end{equation}
  This bound is tighter than $\lab$ bounds of random ensembles \cite{CRT1} , and $\lbb$ of expander-based methods \cite{IR}. 
}
\label{sec:quad}
\section{Resilience to Noise}
\label{sec:noise}
\subsection{Noisy Measurements}
In this Section, we consider deterministic sensing matrices satisfying the hypothesis of Theorem \ref{mainthm}, and show resilience to independent identically distributed (iid) Gaussian noise that is uncorrelated with the measured signal. Note we have introduced the square of $(1\pm\epsilon')$ in \eqref{bnded} merely to simplify the 
{\icd{notation in the}} proof. {\icd{(This $\epsilon'$ could be, for instance, picked so that $\epsilon'(2-\epsilon')\geq \epsilon$, where $\epsilon$ has the same meaning as in 
Theorem \ref{mainthm}.)}}

\begin{theorem}Let $\A$ and $\as$ be such that 
\begin{equation}\label{bnded}
 (1-\epsilon')^2 ||\as||^2 \leq ||\nor \A \as||^2 \leq (1+\epsilon')^2 ||\as||^2{\icd{,}}
\end{equation}
{\icd{with probability exceeding $\delta>0$}}, and let $f=\nor \A\as + \nu$, where the noise samples $\nu(x)$ are iid
complex Gaussian random variables with zero mean and variance $2\sigma^2$. Then, for  
$\gamma\geq 0$,
\begin{equation}\label{noisebnd}
 (1-\epsilon'-\gamma)^2 ||\as||^2 \leq ||f||^2 \leq (1+\epsilon'+\gamma)^2 ||\as||^2,
\end{equation}
with probability greater than
{\icd{$1-2\left(\delta + \mathcal{S}\left[\frac{\gamma \|\as\|}{ \sigma}\right]\right) $,
where 
\[
\mathcal{S}(r)\doteq\left(\int_r^{\infty}e^{-y^2/2} y^{N-1} dy \right) \left(\int_0^{\infty}e^{-y^2/2} y^{N-1} dy \right)^{-1} 
\]
}}
\end{theorem}
\begin{proof}
First consider the probability that $||f||$ exceeds the upper bound in \eqref{noisebnd}.
{\icd{
Setting $g=\nor\Phi\alpha$, we have 
\begin{equation}
\begin{split}\nonumber
&~\quad \text{Pr}\left[\|f\| \geq (1+ \epsilon' +\gamma )\|\as\|\right] \\ &\leq
\text{Pr}\left[\|g\|+ \|\nu\|\geq (1+ \epsilon' +\gamma )\|\as\|\right] \\
&\leq \text{Pr}\left[\|g\|\geq (1+ \epsilon')\|\as\|\right] + 
\text{Pr}\left[\|\nu\|\geq \gamma \|\as\|\right] \\
&\leq \delta + \frac{1}{(2 \pi \sigma^2)^{N/2}} \int_{\|y\|>\gamma\|\as\|}
\exp\left(-\frac{1}{2 \sigma^2} \|y\|^2 \right) d^Ny\\
&=\delta + \frac{1}{(2 \pi)^{N/2}}\int_{\|u\|>\gamma\|\as\|/\sigma}
e^{-\|u\|^2/2} d^Nu
\end{split}
\end{equation}
The estimate for $\text{Pr}\left[\|f\| \leq (1- \epsilon' -\gamma )\|\as\|\right]$
is similar, and the desired bound then follows from the union bounds.
}}
\end{proof}
\subsection{Noisy Signals}
If the signal $\as$  is contaminated by {\icd{white gaussian}} noise then the
measurements are given by
\begin{equation}
 {y} =\nor \left( \Phi {\as} + \Phi {\mu}\right),
\end{equation}
where $\mu$ is complex multivariate Gaussian distributed, with zero mean and covariance 
\begin{equation}
E(\mu\mu^\dagger)=2\sigma^2 I_{\mathcal C\times\mathcal C}.
\end{equation}
The reconstruction algorithm {\icd{ thus}} needs to recover the signal from the noisy measurements
\begin{equation}
 {y} =  {f} + { \nu},
\end{equation}
where ${\nu} = \nor \A {\mu}$ is complex multivariate Gaussian distributed with mean zero and covariance 
\begin{equation}
\mathbb{E}(\nu\nu^\dagger) = \frac{2\sigma^2}{\mm} \Phi\Phi^\dagger.
\end{equation} 
The deterministic compressive sensing schemes considered in this paper have some advantage over 
random compressive sensing schemes in that $\nor\Phi\left(\nor \Phi^\dagger\right) = \frac{\mathcal C}{N} I_{N\times N}$
and consequently $E(\nu(x)\overline{\nu(x')}) =  \frac{2\sigma^2\mathcal C}{N} 
\delta_{x,x'}$, i.e., the noise samples on distinct measurements are independent. 
{\icd{ One can thus use the estimates of the previous subsection again. Noise
of this type is of course harder to deal with; this is illustrated here by the measurement variance being a (possibly huge) factor
$\n/\mm$ larger than the source noise variance $\sigma^2$.
}}

\section{Conclusions}
We have provided simple criteria, that when satisfied by a deterministic sensing matrix, guarantee successful recovery of all but an exponentially small fraction of k-sparse signals. These criteria are satisfied by many families of deterministic sensing matrices including those formed from subcodes of the second order binary Reed Muller codes. The criteria also apply to random Fourier ensembles, where they improve known bounds on the number of measurements required for sparse reconstruction. Our proof of unique reconstruction uses a version of the classical McDiarmid Inequality that may be of independent interest.

We have described a reconstruction algorithm for Reed Muller sensing matrices that takes special advantage of the code structure. Our algorithm requires only vector-vector multiplication in the measurement domain, and as a result, reconstruction complexity is only quadratic in the number of measurements. This improves upon standard reconstruction algorithms such as Basis and Matching Pursuit that require matrix-vector multiplication and have complexity that is superlinear in the dimension of the data domain.
\section*{Acknowledgment}
The authors would like to thank Lorne Applebaum, Richard Baraniuk, Doug Cochran, Ingrid Daubechies, Anna Gilbert, Shamgar Gurevich, Ronnie Hadani, Piotr Indyk, Robert Schapire, Joel Tropp,  Rachel Ward, and the anonymous reviewers for their insights and helpful suggestions.
\bibliographystyle{IEEEbib} 
\bibliography{rss_paper} 
\appendices
\section{Properties of Delsarte-Goethals Sensing Matrices}
First we prove that the columns of the $r^{th}$ Delsarte-Goethals sensing matrix form a group under pointwise multiplication.
\begin{aproposition} Let $\G=\G(m,r)$ be the set of column vectors $\varphi_{P,b}$ where 
$$\varphi_{P,b}(x)=i^{wt(d_P)+2wt(b)} i^{xPx^\top+2bx^\top}~,~\mbox{for }x\in \mathbb{F}_2^m$$
where $b\in \mathbb{F}_2^m$ and where the binary symmetric matrix $P$ varies over the Delsarte-Goethals set $DG(m,r)$. Then $\G$ is a group of order $2^{(r+2)m}$ under pointwise multiplication.
\end{aproposition}
\begin{proof}
 We have
\begin{eqnarray} \nonumber&&\varphi_{P,b}(x)\varphi_{P',b'}(x)\\\nonumber&=&i^{wt(d_P)+wt(d_{P'})+2wt(b\oplus b')} i^{x(P+P')x\top +2(b \oplus b')x^\top}\end{eqnarray}
where $\oplus$ is used to emphasize addition in $\mathbb{F}_2^m$. Write $P+P'=(P \oplus P')+2Q~(mod~4)$ where $Q$ is a binary symmetric matrix. Observe that $2xQx^\top=2d_Qx^\top(mod~4)$, where the diagonal $d_Q=d_P*d_{P'}$ is a pointwise product of $d_P$ and $d_{P'}$. 
\\Thus
\begin{eqnarray}\nonumber &&\varphi_{P,b}(x)\varphi_{P',b'}(x)\\\nonumber&=&i^{\left( [wt(d_P)+wt(d_{P'})+2wt(d_P*d_{P'})]+2wt(b\oplus b'\oplus d_P*d_{P'})\right)}\\\nonumber&~& i^{x(P+P')x^\top +2(b\oplus b'\oplus d_P*d_{P'})x^\top}\\ \nonumber &=&\varphi_{P\oplus P', b\oplus b'\oplus d_P*d_{P'}}(x),
\end{eqnarray}
and $\G$ is closed under pointwise multiplication.
Hence the possible inner products of columns $\varphi_{P,d},\varphi_{P',d'}$ are exactly the possible column sums for columns $\varphi_{Q,b}$ where $Q=P\oplus P'$.\end{proof} Next we verify property (St3). 
\begin{aproposition}
Let $Q$ be a binary symmetric $m\times m$ matrix with rank $r$ and let $b \in \mathbb{F}_2^m$. If $$S=\sum_x i^{xQx^\top+2bx^\top}$$ then either $S=0$ or 
$$S^2=i^{z_1Qz_1^\top+2bz_1^T} 2^{2m-r}~~,~~~\mbox{where }z_1Q=d_Q.$$
\end{aproposition} 
\begin{proof}
We have
\begin{eqnarray}
\nonumber S^2&=&\sum_{x,y} i^{xQx^\top+yQy^\top+2b(x+y)^\top}\\ \nonumber &=& \sum_{x,y}i^{(x+y)Q(x+y)^\top+2*Qy^\top+2b(x+y)^\top}
\end{eqnarray}
Changing variables to $z=x\oplus y$ and $y$ gives
$$S^2=\sum_{z} i^{zQz^\top} \sum_y(-1)^{(d_Q+zQ)y^\top}.$$
Since the diagonal $d_Q$ of a binary symmetric matrix $Q$ is contained in the row space of $Q$ there exists a solution $zQ=d_Q$. The solutions to the equation $zQ=0$ form a vector space $E$ of dimension $m-r$, and for all $e,f \in E$
$$eQe^\top+fQf^\top=(e+f)Q(e+f)^\top~~(mod~4).$$
Hence
\begin{eqnarray}
\nonumber S^2&=&2^m \sum_{e \in E} i^{(z_1+e)Q(z_1+e)^\top +2(z_1+e)b^\top} \\ \nonumber 
&=& 2^m i^{z_1Qz_1^\top+2z_1b^\top} \sum_{e \in E} i^{eQe^\top+2eb^\top}.
\end{eqnarray}
The map $e\rightarrow eQe^\top$ is a linear map from $E$ to $\mathbb{Z}_2$, so the numerator $eQe^\top+2eb^\top$ also determines a linear map from $E$ to  $\mathbb{Z}_2$ (here we identify  $\mathbb{Z}_2$ and  $2\mathbb{Z}_4$). If this linear map is the zero map then
$$S^2=2^{2m-r} i^{z_1 Qz_1^\top+2bz_1^\top},$$
and if it is not zero then $S=0$. Note that given $e \rightarrow eQe^\top$, there are $2^n$ ways to choose $b$ so that $e\rightarrow eQe^\top+2eb^\top$ is the zero map.
\end{proof}
The $0^{th}$ Delsarte-Goethals sensing matrix is a matrix with $\mm=2^m$ rows and $\mm^2$ columns. These columns are the union of $\mm$ mutually unbiased bases, where vectors in one orthogonal basis look like noise to all other orthogonal bases.

\section{Generalized McDiarmid's Inequality}
%\cre{
The method of {\icd{``}}independent bounded differences{\icd{''} (\cite{mcdiarmid}) gives large-deviation concentration bounds for multivariate functions in terms of the maximum 
effect {\icd{on the function value of changing just one coordinate}}. This method has been widely used in combinatorial applications, and in learning theory. In this appendix, we prove that a modification of McDiarmid's inequality also holds for {\icd{\em distinct}} (in contrast to {\icd{\em independent )}} random variables{\icd{; our proof consists again in}} forming martingale sequences.
\\We first introduce some notation. Let $\cx_1,\cdots,\cx_m$ be probability spaces and define $\cx$ as the probability space of all distinct $m$-tuples{\icd{, t}}hat is{\icd{,}} 
\begin{equation}
\label{x}
\cx\doteq\{(x_1,\cdots,x_m)\in \Pi_{i=1}^m \cx_i\mbox{ such that }\forall~i\neq j~:x_i\neq x_j\}.
\end{equation} 
{\icd{(This definition is spelled out in more detail at the end of subsection 3.1.2.)}}
Let $f(x_1,\cdots,x_m)$ be a function from $\cx$ to $\mathbb{R}${\icd{, and let 
$f(X_1,\, \ldots, \,X_m)$ be the corresponding random variable on $\cx$.}} Denote 
{\icd{by $\Xoni$}} the {\icd{$i$- tuple}} of random variables $(X_1,\cdots,X_i)$ on the probability space $\cx$. {\icd{(The ``complete'' $m$-tuple $(X_1,\ldots,X_m)$ will also be denoted by just $X$.)
Analogously, define $\Xim$ to be the $(m-i)$- tuple of random variables 
$(X_{(i+1)},\cdots,X_m)$. We shall also use the notations $\xoni \doteq (x_1, \ldots, x_i) \in \Pi_{\ell=1}^i\cx_{\ell}$, $\cxoni \doteq \{\xoni \in \Pi_{n=1}^i; \,x_{\ell}\neq x_n \mbox{ if } \ell\neq n  \}$; $\xim \in \Pi_{\ell=(i+1)}^m\cx_{\ell}$ and 
$\cxim \subset \Pi_{n=(i+1)}^m$ are
defined analogously. \\
}}
\begin{atheorem}[{\icd{Self-avoiding}} McDiarmid inequality]
Let $\cx$ be the probability space defined in Equation~(\mbox{\ref{x}}), and let $f:\cx\rightarrow\mathbb{R}$ be a function such that for any {\icd{index}} $i$, {\icd{and any}} ${\icd{\xonim 
\in \cxonim,}}$
\begin{eqnarray}
&&\hskip-0.5cm\sup_{u\in \cx_i; u \neq x_{n}, n=1\rightarrow i}\Ex[f(x_1,\cdots,x_{i-1},u,X_{i+1},\cdots,X_m)]\quad\quad\quad\quad\label{concentration}\\
&&\hskip-2.0cm\quad\quad\quad\quad-\,\inf_{l\in \cx_i; l \neq x_{n}, n=1\rightarrow i}\Ex[f(x_1,\cdots,x_{i-1},l,X_{i+1},\cdots,X_m)\leq c_i.\nonumber
\end{eqnarray}
Then for any positive $\epsilon$,
\begin{equation}
\label{result}
%\Pr\left[{\icd{\left|}}f(X_1,\cdots,X_m){\icd{-}}\Ex[f(X_1,\cdots,X_m)]{\icd{\right| \geq}}\epsilon\right]\leq 2\exp\left(\frac{-2\epsilon^2}{\sum c_i^2}\right).
\hskip-0,5cm\Pr\left[\left|f(X_1,\cdots,X_m)-\Ex[f(X_1,\cdots,X_m)]\right| \geq\epsilon\right]\leq 2\exp\left(\frac{-2\epsilon^2}{\sum c_i^2}\right).
\end{equation}
\end{atheorem}

{\icd{Our proof will invoke}} Hoeffding's Lemma [\cite{mcdiarmid}]
\begin{aproposition}[Hoeffding's Lemma] Let $X$ be a random variable with $\Ex[X]=0$ and $a\leq X\leq b$ then for $t>0$
 $$\Ex\left[ e^{tX}\right] \leq\exp\left\{ \frac{t^2(b-a)^2}{8}\right\}.$$
 \end{aproposition}

{\icd{In our proof we will also make use of the functions}}
$$
{\icd{Z_i(\xoni)\doteq \Ex[f(X)|\Xoni=\xoni] \quad\mbox{ where }~ \xoni \in \cxoni}}
$$
{\icd{As a result, for all $\xonim$ in $\cxonim\,$
\begin{eqnarray}
\hskip-1.2cm\,\left|\,\sup_{u\in \cx_i; u \neq x_{n}, n=1\rightarrow i}
Z_i(\xonim,u) \,-\, \inf_{l\in \cx_i; l \neq x_{n}, n=1\rightarrow i}Z_i(\xonim,l)\,\right|
\end{eqnarray}
is less than $c_i$. This implies, for all $\xoni \in \cxoni$, 
\begin{eqnarray} \nonumber-c_i &\leq& \inf_{l\in \cx_i; l \neq x_{n}, n=1\rightarrow i-1}Z_i(\xonim,l)\\\nonumber&-&\sup_{u\in \cx_i; u \neq x_{n}, n=1\rightarrow i-1}
Z_i(\xonim,u)\nonumber\\
&\leq& Z_(\xoni) \\\nonumber&-& \Ex[f(\Xonim,X_i,\Xim)|\xonim] \nonumber\\
&=& Z_i(\xoni) - Z_{i-1}(\xonim)\nonumber\\
&\leq& Z_i(\xoni) \\\nonumber&-& \inf_{l\in \cx_i; l \neq x_{n}, n=1\rightarrow i-1}Z_i(\xonim,l)\nonumber\\
&\leq&\sup_{u\in \cx_i; u \neq x_{n}, n=1\rightarrow i-1}
Z_i(\xonim,u) \,\\\nonumber&-&\, \inf_{l\in \cx_i; l \neq x_{n}, n=1\rightarrow i-1}Z_i(\xonim,l)\nonumber\\
&\leq& c_i~, \nonumber
\end{eqnarray}
or
\begin{equation}
 \left| Z_i(\xoni) - Z_{i-1}(\xonim)  \right| \leq c_i 
\label{z-concentration} 
\end{equation}
{\icd{Until now, we have viewed each $Z_i$ as a function on the subset $\cxoni$ of $\Pi_{\ell=1}^i \cx_{\ell}$; it is straightforward to lift the $Z_i$ to functions on all of $\cx$.
The $Z_i(\Xoni)= Z_i(X)$ can also be considered as random variables on $\cx$, depending only on the first $i$ components of $X$,}}
$$
{\icd{Z_i(\Xoni)=\Ex_{\Xim}[f(X)|\Xoni] }}
$$
{\icd{(The subscript $\Xim$ on the expectation indicates that one averages only with respect to the variables listed in the subscript, in this case the last $m-i$ variables. We adopt this subscript convention in what follows; only expectations without subscript are with respect to
the whole probability space $\cx$.) }}\\
{\icd{Viewing the $Z_i$ as random variables,  we o}}bserve that $Z_0=\Ex[f(X_1,\cdots,X_m)]$, and {\icd{that}} $Z_{m}=f({\icd{X}}_1,\cdots,{\icd{X}}_m)$. 
{\icd{Because of the restriction to $\cx$, the random variables $X_{\ell}$, $Z_{\ell}$  are not independent. However, with respect to averaging over $X_i$, the $Z_i,\,i=1,\,\ldots,\,m$ constitue a martingale in the following sense:}} 
\begin{equation}
\label{martingale}
{\icd{\Ex_{X_i}[Z_i(X)|\Xonim]=Z_{i-1}(X)~,}}
\end{equation}
%{\icd{which is equivalent to
%\begin{eqnarray}
%\label{funct_martingale}
%&&\forall \xonim \in \cxonim\,:\\\nonumber&&\, \Ex_{X_i}[Z_i(\xonim,X_i)]=Z_{i-1}(\xonim)~.
%\label{funct_martingale}
%\end{eqnarray} 
%}}}}
\begin{proof} 
Using Markov's inequality, we see that for any positive {\icd{$t$}}
 \begin{eqnarray}
 \Pr\left[f- \Ex[f]\geq\epsilon\right] &=&\Pr\left[e^{t(f-\Ex[f])}\geq e^{t\epsilon}\right]\nonumber\\  &\leq&e^{-\epsilon t}\,\Ex\left[e^{t(f-\Ex[f])}\right]\label{markov}
 \end{eqnarray}
{\icd{Since $f-\Ex[f]=Z_m-Z_0$, we can rewrite this as}}
\begin{equation*}
%\label{hard} 
\Ex\left[e^{t(f-\Ex[f])}\right]= \Ex\left[ \exp\left(t\sum_{i=1}^m {\icd{(}}Z_i-Z_{i-1}{\icd{)}}\right)\right]
\end{equation*}
By marginalization of {\icd{the}} expectation, 
\begin{eqnarray}
&&\nonumber   \Ex\left[ \exp \left(t\sum_{i=1}^m {\icd{(}}Z_i-Z_{i-1}{\icd{)}} \right)\right]\\&=&   
\Ex_{\icd{X_{1\rightarrow(m-1)}}}
\left[ \Ex_{\icd{X_m}}\left[\exp \left(t\sum_{i=1}^m {\icd{(}}Z_i-Z_{i-1}{\icd{)}}\right)\left| {\icd{X_{1\rightarrow(m-1)}}}\right.\right]\right]
 \nonumber  \\ \nonumber &=&  \Ex\left[ \exp \left(\sum_{i=1}^{m-1} {\icd{(}}Z_i-Z_{i-1}{\icd{)}}\right)\Ex_{\icd{X_m}}\left[e^{t(Z_m-Z_{m-1})}\left| {\icd{X_{1\rightarrow(m-1)}}}\right.\right]\right]{\icd{~,}}
 \end{eqnarray}
{\icd{where we have used that each $Z_i$ depends on only the first $i$ components of $X$, so that only $(Z_{m-1}-Z_m) $ is affected by the averaging over $X_m$.\\
By (\ref{z-concentration}), we have, for all $\xoni \in \cxoni$, 
$|Z_i(\xoni)-Z_{i-1}(\xonim) | \leq c_i$, which can also be rewritten as 
$-c_i \leq Z_i(X)-Z_{i-1}(X) \leq c_i$.\\
Because of the martingale property (\ref{martingale}) we have
$\Ex\left[Z_i-Z_{i-1} | X_1^{i-1}\right]=\Ex_{\Xim}\left[Z_i-Z_{i-1} | X_1^{i-1}\right]
=\Ex_{X_i}\left[Z_i-Z_{i-1} | X_1^{i-1}\right]=0$.\\
Combining these last two observations with
Hoeffding's Lemma \cite{mcdiarmid}  we conclude }}
\begin{eqnarray}%\label{good}
&&\!\!\!\!\!\!\Ex_{X_{1 \rightarrow m}}\left[ \exp\left(\sum_{i=1}^{m} (Z_i-Z_{i-1})\,\right)\right]\nonumber\\\nonumber &=& \hskip-0.1cm\Ex_{X_{1 \rightarrow (m-1)}}\left[\, e^{\left(\,t\,\sum_{i=1}^{m-1} (Z_i-Z_{i-1})\,\right)}
\,\Ex_{X_m}\left[ e^{(Z_m-Z_{m-1})}\left| X_{1 \rightarrow (m-1)}\right.\right]\,\right]\nonumber\\ 
&\leq& \hskip-0.1cme^{t^2c_m^2/8}~ \Ex_{X_{1 \rightarrow (m-1)}}\left[ \exp\left(\sum_{i=1}^{m-1} (Z_i-Z_{i-1})\,\right)\right]\nonumber\\\nonumber &\leq& \, \cdots \,
 \leq \exp \left(\,\frac{1}{8}\,t^2\,\sum_{i=1}^m\,c_i^2\,\right)\nonumber
 \end{eqnarray} 
Substituting this into~(\ref{markov}) we obtain
\begin{equation}\label{allt}
\Pr\left[f- \Ex[f]\geq\epsilon\right]\leq \exp \left(\,-t\,\epsilon\,+\,\frac{1}{8}\,t^2\,\sum_{i=1}^m\,c_i^2\,\right)
\end{equation}
Since equation~(\ref{allt}) is valid for any $t>0$, we can optimize over $t$. By substituting the value $t=4\epsilon\,(\,\sum c_i^2\,)^{-1}$  we get 
\begin{equation*}
\Pr\left[f- \Ex[f]\geq\epsilon\right]\leq \exp\left(\frac{-2\epsilon^2}{\sum c_i^2}\right).
\end{equation*}
Replacing the function $f$ by $\Ex[f]-f$, it follows that $\Pr\left[ f-\Ex[f] \leq-\epsilon\right]\leq \exp\left(\frac{-2\epsilon^2}{\sum c_i^2}\right)$;  union bounds therefore imply that
$$\Pr\left[| f-\Ex[f]| \geq\epsilon\right]\leq 2\exp\left(\frac{-2\epsilon^2}{\sum c_i^2}\right).$$\vskip-1cm
\end{proof}
%}

%\input{app3}
\end{document}